\documentclass[11pt,letterpaper]{article}
\usepackage{geo}

\newcommand{\kCSP}{k\mathrm{CSP}}
\newcommand{\CSP}{\mathrm{CSP}}
\newcommand{\kSAT}{k\mathsf{SAT}}
\newcommand{\kXOR}{k\mathsf{XOR}}
\newcommand{\ThreeSAT}{3\mathsf{SAT}}
\newcommand{\ThreeXOR}{3\mathsf{XOR}}
\newcommand{\TwoXOR}{2\mathsf{XOR}}
\newcommand{\tXOR}{t\mathsf{XOR}}
\newcommand{\kmoXOR}{(k-1)\mathsf{XOR}}
\newcommand{\XOR}{\mathsf{XOR}}
\newcommand{\SAT}{\mathsf{SAT}}
\newcommand{\ER}{Erd\H{o}s-R\'{e}nyi\xspace}
\newcommand{\Lovasz}{Lov\'{a}sz\xspace}
\newcommand{\Inst}{\calI}
\newcommand{\bInst}{\bcalI}
\newcommand{\pInst}{\calI_+} 
\newcommand{\HDist}[3]{\calH_{#1}^{#3}(#2)}
\newcommand{\UnsignedH}[3]{\calH_{#1,+}^{#3}(#2)}
\newcommand{\cube}{\left\{\pm \frac{1}{\sqrt{n}} \right\}^n}
\newcommand{\spn}{\operatorname{span}}
\newcommand{\G}{\bM} 
\newcommand{\primalH}{H_P} 
\newcommand{\bprimalH}{\bH_P} 
\newcommand{\GOE}{\mathsf{GOE}}
\newcommand{\OPT}{\mathsf{OPT}}
\newcommand{\cls}{U} 
\newcommand{\T}{T} 
\newcommand{\Binom}{\mathrm{Binom}}
\newcommand{\meanA}{\ol{A}} 
\newcommand{\meanbA}{\ol{\bA}}
\newcommand{\Ram}{r_d} 
\newcommand{\Cindset}{C_d} 
\newcommand{\Cynorm}{c_d} 
\newcommand{\indsets}[1]{\calS_\eta(#1)} 
\newcommand{\Y}[1]{\calY_\eta(#1)} 
\newcommand{\TruncInst}[2]{\left.\Inst\right|_{#1,#2}}
\newcommand{\TruncbInst}[2]{\left.\bInst\right|_{#1,#2}}
\newcommand{\poseps}{\eps_{\ref{lem:induced-two-xor-balanced}}}

\begin{document}
\title{Certifying solution geometry in random CSPs: counts, clusters and balance}
\author{Jun-Ting Hsieh\thanks{Carnegie Mellon University. \texttt{juntingh@cs.cmu.edu}.  Supported by NSF CAREER Award \#2047933.} \and Sidhanth Mohanty\thanks{University of California at Berkeley.  \texttt{sidhanthm@cs.berkeley.edu}.  Supported by Google PhD Fellowship.} \and Jeff Xu\thanks{Carnegie Mellon University. \texttt{jeffxusichao@cmu.edu}.  Supported by NSF CAREER Award \#2047933.}}
\date{\today}
\maketitle

\begin{abstract}
	An active topic in the study of random constraint satisfaction problems (CSPs) is the geometry of the space of satisfying or almost satisfying assignments as the function of the density, for which a precise landscape of predictions has been made via statistical physics-based heuristics.  In parallel, there has been a recent flurry of work on \emph{refuting} random constraint satisfaction problems, via nailing refutation thresholds for spectral and semidefinite programming-based algorithms, and also on \emph{counting} solutions to CSPs.  Inspired by this, the starting point for our work is the following question:
	\begin{displayquote} 
		\emph{What does the solution space for a random CSP look like to an efficient algorithm?}
	\end{displayquote}
	In pursuit of this inquiry, we focus on the following problems about random Boolean CSPs at the densities where they are unsatisfiable but no refutation algorithm is known.
	\begin{enumerate}
		\item {\bf Counts.} For every Boolean CSP we give algorithms that with high probability certify a subexponential upper bound on the number of solutions.  We also give algorithms to certify a bound on the number of large cuts in a Gaussian-weighted graph, and the number of large independent sets in a random $d$-regular graph.
		\item {\bf Clusters.}  For Boolean $3$CSPs we give algorithms that with high probability certify an upper bound on the number of \emph{clusters} of solutions.
		\item {\bf Balance.} We also give algorithms that with high probability certify that there are no ``unbalanced'' solutions, i.e., solutions where the fraction of $+1$s deviates significantly from $50\%$.
	\end{enumerate}
	Finally, we also provide hardness evidence suggesting that our algorithms for counting are optimal.
\end{abstract}

\thispagestyle{empty}
\setcounter{page}{0}
\newpage
\tableofcontents
\thispagestyle{empty}
\setcounter{page}{0}
\newpage

\section{Introduction}
Constraint satisfaction problems (CSPs) are fundamental in the study of algorithm design and complexity theory.  They are simultaneously simple and also richly expressive in capturing a wide range of computational tasks, which has led to fruitful connections to other areas of theoretical computer science (see, for example, \cite{Gol11,ABW10} for connections to cryptography, \cite{DLSS14} for applications to hardness of learning, and \cite{Fei02} for applications to average-case hardness).  Hence, understanding them has received intense attention in the past few decades, leading to several comprehensive theories of their complexity.  Some of the highlights include: the Dichotomy Theorem, which characterizes the worst-case complexity of satisfiability of CSPs via their algebraic properties \cite{Sch78,BJK05,Zhuk20}, inapproximability results via the PCP Theorem \cite{Has01}, and the theory of optimal inapproximability based on connections between semidefinite programming and the Unique Games conjecture \cite{Kho02,KKMO07,Rag08}.

In this work, we are interested in the algorithmic aspects of random instances of CSPs.  There has been a diverse array of phenomena about random CSPs illustrated in recent work, of dramatically varying nature depending on the ratio of the number of constraints to the number of variables, known as the \emph{density}.  Of central importance is the \emph{satisfiability threshold}, which marks a phase transition where a random CSP instance shifts from being likely satisfiable to being likely unsatisfiable.  When the density is well below the satisfiability threshold, there are several algorithms for tasks such as counting and sampling assignments to a random CSP instance \cite{Moi19,JPV21,GGGY19,BGG19}, whereas well above this threshold there are efficient algorithms for \emph{certifying} that random CSPs are unsatisfiable \cite{AOW15}.  The densities in the interim hold mysteries that we don't yet fully understand, and this work is an effort to understand the algorithmic terrain there.  To make matters concrete, for now we will specialize the discussion of the problem setup and our work to the canonical $\ThreeSAT$ predicate.

Consider a random $\ThreeSAT$ formula $\bInst$ on $n$ variables and $\Delta n$ clauses where each clause is sampled uniformly, independently, and adorned with uniformly random negations.  Once the density $\Delta$ is a large enough constant, this random instance is unsatisfiable with high probability.\footnote{In fact, it is conjectured that there is a sharp threshold for unsatisfiability once $\Delta$ crosses some constant $\alpha_{\mathrm{SAT}}\approx 4.267$.}  On the other hand, the widely believed Feige's random $3$SAT hypothesis \cite{Fei02} conjectures that when $\Delta$ is any constant, there is no algorithm to \textit{certify} that a random instance is unsatisfiable.  Further, the best known algorithms for efficiently certifying that it is unsatisfiable require $\Delta\gtrsim\sqrt{n}$ \cite{GL03,COGL07,FO07,AOW15}.  Moreover, when $\Delta\lesssim\sqrt{n}$ there is a lower bound against the Sum-of-Squares hierarchy \cite{Gri01,Sch08} (known to capture many algorithmic techniques), which suggests an \emph{information-computation gap} and earns $\sqrt{n}$ the name \emph{refutation threshold}.

In this picture, at both densities $n^{.25}$ and $n^{.35}$, $\bInst$ is likely unsatisfiable but ``looks'' satisfiable to an efficient algorithm.
But is there a concrete sense in which a random formula at density $n^{.25}$ is ``more satisfiable'' than one at density $n^{.35}$ from the lens of a polynomial-time algorithm?  A natural measure of a $\ThreeSAT$ formula's satisfiability is its number of satisfying assignments, which motivates the following question.
\begin{displayquote}
	\emph{What is the best efficiently certifiable upper bound on the number of assignments satisfying $\bInst$?}
\end{displayquote}
In the context of $\ThreeSAT$, our work proves:
\begin{theorem}[Informal]	\label{thm:main-thm-3SAT}
    There is an efficient algorithm to certify with high probability that a random $\ThreeSAT$ formula with density $\Delta = n^{1/2-\delta}$ has at most $\exp(\wt{O}(n^{3/4+\delta/2}))$ satisfying assignments.
\end{theorem}
In addition to certifying the number of satisfying assignments, we can certify that the solutions form clusters and upper bound the number of clusters under the refutation threshold.

\paragraph{Clusters.}
Besides the satisfiability threshold, random $\kSAT$ is conjectured to go through other phase transitions too, as predicted in the work of \cite{KMRSZ07}.  In particular, the \textit{clustering threshold} is the density where the solution space is predicted to change from having one giant component to roughly resembling a union of several small Hamming balls, known as \emph{clusters}, that are pairwise far apart in Hamming distance.

Much like the refutation threshold that marks where efficient algorithms can witness unsatisfiability, it is natural to ask if there is some regime under the refutation threshold where an efficient algorithm can witness a bound on the number of clusters of solutions.  The following more nuanced version of \pref{thm:main-thm-3SAT} gives an answer to this question.
\begin{theorem}[Informal]	\label{thm:main-thm-clustering-3SAT}
    There is an efficient algorithm to certify with high probability that the satisfying assignments of a random $\ThreeSAT$ formula with density $\Delta = n^{1/2-\delta}$ are covered by at most $\exp(\wt{O}(n^{1/2+\delta}))$ diameter-$\wt{O}(n^{3/4+\delta/2})$ clusters.
\end{theorem}

\paragraph{Balance in the solution space.}  Suppose at density $\Delta$, a typical $\ThreeSAT$ formula has $\sim\exp(c_{\Delta} n)$ satisfying assignments, then due to the uniformly random negations in clauses, each string is satisfying with probability $\sim\exp((c_{\Delta}-1)n)$.  Then one can show via the first moment method that with high probability there are no satisfying assignments with Hamming weight outside $\left[\frac{1}{2}-f(c_{\Delta}), \frac{1}{2}+f(c_{\Delta})\right].$\footnote{where $f$ is chosen so that the number of strings outside that Hamming range is $\ll\exp((c_{\Delta}-1)n)$.}  In particular, the intersection of the solution space with the set of unbalanced strings empties out under the satisfiability threshold.  This raises the question:
\begin{displayquote}
    \emph{Is there an efficient algorithm to certify that a random CSP instance has no unbalanced assignments at density significantly under the refutation threshold?}
\end{displayquote}
We affirmatively answer this question and in the special case of $\ThreeSAT$ prove:
\begin{theorem}[Informal] \label{thm:main-thm-global-card-3SAT}
    There is an efficient algorithm to certify with high probability that a random $\ThreeSAT$ formula with density $\Delta = n^{1/2-\delta}$ has no satisfying assignments with Hamming weight outside
    \[
        \left[\frac{1}{2}-\wt{\Theta}\left(\frac{1}{n^{1/4-\delta/2}}\right), \frac{1}{2}+\wt{\Theta}\left(\frac{1}{n^{1/4-\delta/2}}\right)\right].
    \]
\end{theorem} 

We illustrate our upper bounds for counting satisfying assignments and clusters in \pref{fig:plot-of-results}.  We delve into the precise technical statements of our results and the techniques involved in proving them in \pref{sec:our_contributions}.
Then to put our work in context, we survey and discuss existing work on information-computation gaps, and algorithmic work on counting, sampling and estimating partition functions in \pref{sec:related_work}.

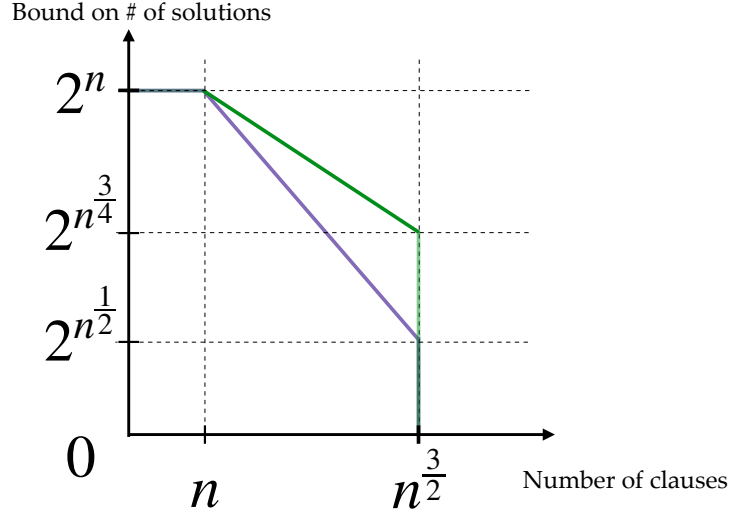
\myfig{.6}{3SAT-fig}{Our results for $\ThreeSAT$. {\bf \color{ForestGreen} Green}: certifiable upper bound on the number of satisfying assignments. {\bf \color{Fuchsia} Purple}: upper bound on the number of clusters of satisfying assignments.  In the case of $\kSAT$, the green plot looks identical but with $n$ replaced by $n^{(k-1)/2}$ and $n^{3/2}$ replaced by $n^{k/2}$.}{fig:plot-of-results}

\subsection{Our Contributions} \label{sec:our_contributions}
In this section, we give a more detailed technical description of our contributions.  To set the stage for doing so, we first formally clarify the notion of \emph{certification} and some preliminaries on constraint satisfaction problems.

Fix a sample space $\Omega$, a probability distribution $\calD$ over $\Omega$, and a function $f: \Omega \to \R$.  For example, \ $\Omega$ is the space of $\ThreeSAT$ instances, $\calD$ is the distribution of instances given by the random $\ThreeSAT$ model, and $f$ is the number of satisfying assignments.

\begin{definition} \label{def:certification_algorithm}
    We say that a deterministic algorithm $\calA$ certifies that $f \leq C$ with probability over $1-p$ over $\calD$ if $\calA$ satisfies
    \begin{enumerate}
        \item For all $\omega \in \Omega$, $f(\omega) \leq \calA(\omega)$.
        \item For a random sample $\bomega \sim \calD$, $\calA(\bomega) \leq C$ with probability over $1-p$.
    \end{enumerate}
\end{definition}
We emphasize that an algorithm that always outputs the typical value of $f$ is \textit{not} a certification algorithm: it will satisfy the second condition but not the first.  
Thus, in several average-case problems, there are gaps between the typical value and the best known certifiable upper bound.

\begin{remark}
    Due to the guarantees of $\calA$, one can think of the ``transcript'' of the algorithm on input $\omega$ as being a proof that $f(\omega)\le\calA(\omega)$.
\end{remark}

\begin{definition}
    A \emph{predicate} $P:\{\pm1\}^k\to\{0,1\}$ is any Boolean function that is not the constant function that always evaluates to $1$.  An instance $\calI$ of a \emph{constraint satisfaction problem} on predicate $P$ and vertex set $[n]$ is a collection of \emph{clauses}, where a clause is a pair $(c,S)$ with $c\in\{\pm1\}^k$ and $S\in[n]^k$.  Given $x\in\{\pm1\}^n$, the value of $\calI$ on $x$ is:
    \[
        \calI(x) \coloneqq \frac{1}{|\calI|}\sum_{(c,S)\in\calI} P(c_1 x_{S_1},\dots, c_k x_{S_k}).
    \]
    We say $x$ \emph{satisfies} a clause $(c,S)$ if $P(c_1 x_{S_1},\dots,c_k x_{S_k})=1$, and say $x$ is \emph{$(1-\eta)$-satisfying} if $\calI(x)\ge1-\eta$.  If $\eta = 0$, we say $x$ is \emph{exactly satisfying}.
\end{definition}

In this work we are concerned with random CSPs. We defer an exact description of the random model to \pref{sec:random-hypergraphs} (note however that the common random models used in the literature are all qualitatively similar; cf.\ \cite[Appendix D]{AOW15}).
Our first result is an algorithm certifying a subexponential upper bound on the number of $(1-\eta)$-satisfying assignments for random CSPs.  
\begin{theorem} \label{thm:kCSP_main_thm}
    Let $\bInst$ be a random $\kCSP$ instance on any predicate $P$ on $n$ variables and $\Delta n$ clauses.
    For every $\eps > 0$, there is an algorithm that certifies with high probability that the number of $(1-\eta)$-satisfying assignments to $\bInst$ is upper bounded by:
    \[
        \exp\left(\wt{O}(\eta n)\right)\cdot\exp\left(\wt{O}\left(\sqrt{\frac{n^{(k+1)/2}}{\Delta}}\right)\right) \cdot \exp\left(O\left(\frac{n^{1+\eps}}{\Delta^{1/(k-2)}}\right)\right).
    \]
\end{theorem}

To more easily parse the statement, let's plug in concrete parameters.
\begin{remark}
    Let's fix the predicate to be $\kSAT$ for any $k\ge 3$, $\eta = 0$, and $\Delta = n^{k/2-1.1}$.  The quantity of interest is the number of exactly satisfying solutions to a random $\kSAT$ formula at a density strictly smaller than the refutation threshold of $\wt{\Omega}(n^{k/2-1})$.  Then, we get an algorithm that with high probability certifies that the number of exactly satisfying assignments is at most:
    \[
        \exp(\wt{O}(n^{0.8})),
    \]
    which is a subexponential bound.  More generally, our algorithms certify a subexponential bound on the number of satisfying assignments for $\kSAT$ for $\Delta = n^{k/2-1.5+c}$ for any $c > 0$ and this bound improves as we increase $c$.
\end{remark}

The proof of \pref{thm:kCSP_main_thm} relies on $3$ ingredients of increasing complexity.  The first is the simple observation that given a $\kCSP$ instance $\calI$ on any predicate $P$, there is a transformation to a $\kSAT$ instance $\calI'$ such that:
\begin{enumerate}[(i)]
    \item For any $\eta>0$, if $x$ is $(1-\eta)$-satisfying for $\calI$, then it is also $(1-\eta)$-satisfying for $\calI'$.
    \item If $\bInst$ is a random instance of a CSP on $P$ with density $\Delta$, then $\bInst'$ is a random instance of $\kSAT$ with density $\Delta$.
\end{enumerate}
This reduction is described in the proof of \pref{cor:kCSP-to-kSAT}.

The second ingredient is a generalization of the ``$\ThreeXOR$-principle'' of \cite{Fei02,FO07}, which we call the ``$\kXOR$-principle''.  The $\kXOR$ principle, which we state below, reduces count certification/refutation for a random $\kSAT$ formula to the same task on a random $\kXOR$ formula.
\begin{lemma}   \label{lem:intro-kXOR-principle}
    Let $\bInst$ be a random $\kSAT$ formula on $m = \Delta n$ clauses.  There is an efficient algorithm that with high probability certifies that any $(1-\eta)$-satisfying assignment of $\bInst$ must $\kXOR$-satisfy at least $\left(1-O(\eta)-\wt{O}\left(\sqrt{\frac{n^{(k-3)/2}}{\Delta}}\right)\right)m$ clauses.
\end{lemma}
We detail the proof in \pref{sec:kSAT}, which is close to the reduction from generic CSP refutation to $\kXOR$ refutation in \cite{AOW15} based on the Fourier expansion.

For the sake of a notationally simple sketch, let's restrict ourselves to the case $\eta = 0$.  We can write $\kSAT(x_1,\dots,x_k) = (1-2^{-k}) + 2^{-k}x_1 x_2\cdots x_k + q(x_1,\dots,x_k)$ where $q$ is a degree-$(k-1)$ polynomial without a constant term.  Thus, given a random $\kSAT$ instance $\bInst$ and any satisfying assignment $x$:
\[
    1 = \bInst(x) = 1-2^{-k} + 2^{-k} \frac{1}{|\bInst|}\sum_{(c,S)\in\bInst}\prod_{i=1}^k c_ix_{S_i} + \frac{1}{|\bInst|}\sum_{c,S\in\bInst} q(c_1x_{S_1},\dots,c_kx_{S_k}).
\]
Once $\Delta\gtrsim n^{(k-3)/2}$ the refutation algorithm of \cite{AOW15} can be employed to certify that the last term is insignificantly small by virtue of the last term being a degree-$(k-1)$ polynomial with no constant term.  This would force $2^{-k} \frac{1}{|\bInst|}\sum_{(c,S)\in\bInst}\prod_{i=1}^k c_ix_{S_i}$ to be near $1$, which is the same as saying $x$ must $\kXOR$-satisfy most clauses.

Our third ingredient for \pref{thm:kCSP_main_thm} is a count certification algorithm for $\kXOR$, which we prove in \pref{sec:kXOR}.
\begin{theorem}
\label{thm:kxor_main_thm}
    For constant $k\geq 3$, consider a random $\kXOR$ instance with $n$ variables and $\Delta n$ clauses.
    For any constant $\eps > 0$, there is a polynomial-time algorithm that certifies with high probability that the number of $(1-\eta)$-satisfying assignments is at most
    \[
        \exp\left(\wt{O}(\eta n)\right)\cdot\exp\left(O\left(\frac{n^{1+\eps}}{\Delta^{1/(k-2)}}\right)\right).
    \]
\end{theorem}
In fact, the certification algorithm only depends on the hypergraph structure of the $\kXOR$ instance and not the signings of each clause.  This is crucial since our algorithm recursively looks at $(k-1)\XOR$ subinstances with unknown signings.  The stronger statement we prove is:
\begin{theorem} \label{thm:boost-kxor_main_thm}
    For constant $k\ge 2$, consider a random $k$-uniform hypergraph $\bH$ on $n$ vertices and $\Delta n$ hyperedges where $\Delta\gg\log n$.  For $\eps > 0$, there is a polynomial-time algorithm that certifies with high probability that the number of $(1-\eta)$-satisfying assignments to any $\kXOR$ instance on $\bH$ is at most
    \[
        \exp\left(\wt{O}(\eta n)\right)\cdot
        \begin{cases}
            1 & \text{if $k=2$} \\
            \exp\left(\wt{O}\left(\frac{n^{1+\eps}}{\Delta^{1/(k-2)}}\right)\right) & \text{if $k\ge 3$}.
        \end{cases}
    \]
\end{theorem}
\pref{thm:boost-kxor_main_thm} is of interest beyond algorithmic considerations as it gives a high-probability bound on the number of approximate solutions for \emph{any} $\kXOR$ formula on a random hypergraph.

\begin{remark}
    Gaussian elimination is able to count exact solutions to an explicit $\kXOR$ instance but fails for counting $(1-\eta)$-satisfying assignments or when the signings are unknown.
\end{remark}

We now give a brief sketch of our proof of \pref{thm:boost-kxor_main_thm}.  Given a random $k$-uniform hypergraph, we would like to certify that \emph{any} $\kXOR$ instance on this hypergraph has no more than $\exp\left(\wt{O}(\eta n)\right)\cdot\exp\left(\wt{O}\left(\frac{n^{1+\eps}}{\Delta^{1/(k-2)}}\right)\right)$ approximate solutions.   We will first present an overview in the context of $\TwoXOR$ as the ``base case'', and then explain the algorithm for $\ThreeXOR$ to illustrate the ``recursive step''.

\paragraph{$\twoxor$ sketch.} Let's consider a random graph $\bG$ on $n$ vertices and $\Delta n$ edges where $\Delta\gg\log n$.  Then, its degrees concentrate and its normalized Laplacian has a large spectral gap (more precisely, a spectral gap of $1-O\left(\frac{1}{\sqrt{\Delta}}\right)$).  As a consequence of Cheeger's inequality, any set $S$ containing fewer than half the vertices has roughly half its edges leaving --- which quantitatively would be around $\Delta|S|$.  We prove that a large spectral gap and concentration of degrees is all that is necessary for any $\twoxor$ instance to have an appropriately bounded number of satisfying assignments.

Now let $\Inst$ be any $\TwoXOR$ instance on $\bG$.  The key observation is that if $x$ and $x'$ are two $(1-\eta)$-satisfying assignments for $\Inst$, then the pointwise product $y\coloneqq x\circ x'$ is $(1-2\eta)$-satisfying for $\Inst_+$, the $\twoxor$ instance on $\bG$ obtained by setting the sign on all constraints to $+1$.  The constraints violated by $y$ are the ones on the cut between $S_+$ and $S_-$, the positive and negative vertices in $y$ respectively.  There are roughly $\Delta\cdot\min\{|S_+|,|S_-|\}$, and consequently $\min\{|S_+|,|S_-|\}\le 2\eta n$ since $y$ is $(1-2\eta)$-satisfying.  In particular, $y$ either has at most $2\eta n$ positive entries or $2\eta n$ negative entries.  The upshot is the number of $(1-\eta)$-satisfying assignments of $\Inst$ is at most $\exp(\wt{O}(\eta n))$.  This sketched argument is carefully carried out in \pref{sec:2XOR}.

\paragraph{$\ThreeXOR$ sketch.}  Now, let $\bH$ be a random hypergraph on $n$ vertices and $\Delta n$ hyperedges.  The observation here is that for any $\ThreeXOR$ instance $\Inst$ on $\bH$, any assignment $x$ that $(1-\eta)$-satisfies $\Inst$ also approximately satisfies a particular \emph{induced $\TwoXOR$ instance} of a fixed subset of variables $S$ (cf.\ \pref{def:induced_txor}).  The induced $\TwoXOR$ instance's underlying graph $G$ is fixed and distributed like a random graph, and only the signings on the edges vary as we vary $x$.   That lets us run the algorithm for $\TwoXOR$ on $G$ to obtain an upper bound $F$ on all induced instances on $G$, which then yields a bound of $2^{|S|}\cdot F$.  This is where we use that our algorithm depends only on the underlying graph, hence avoiding an enumeration of all assignments to variables in $S$.

We immediately see that for a fixed subset $S$, the above procedure throws away most of the clauses (keeping only clauses that have 1 variable in $S$).
Thus, it is clearly suboptimal to look at just one subset $S$.
To resolve this, we partition the variables into subsets $S_1,\dots, S_\ell$, run the algorithm on each of them, and aggregate the results.  This is explained in detail in the proofs of \pref{lem:average_partition} and \pref{thm:kxor_upper_bound}.

\paragraph{Clustering.}  Our next result is an algorithm to upper bound the number of clusters formed by the solutions.  Given $x\in\{\pm1\}^n$, we call the Hamming ball $B(x,r)$ a \emph{radius-$r$ cluster} or \emph{diameter-$2r$ cluster}.  For $3\CSP$s we prove in \pref{cor:3SAT_clusters}:
\begin{theorem}
    Let $P$ be any $3$-ary predicate, and let $\bInst$ be a random instance of $P$ on $n$ variables and $\Delta n$ clauses.
    Let $\eta \in [0,\eta_0]$ where $\eta_0$ is a universal constant, and let $\theta \coloneqq 8\eta + O\left(\sqrt{\frac{\log^5 n}{\Delta}}\right)$.
    There is an algorithm that certifies with high probability that the $(1-\eta)$-satisfying assignments to $\bInst$ as a $P$-$\CSP$ instance are covered by at most
    \begin{equation*}
        \exp(O(\theta^2 \log(1/\theta))n)
    \end{equation*}
    diameter-$(\theta n)$ clusters.
\end{theorem}

Inspecting the proof of counting $\TwoXOR$ (specifically the argument about $\pInst$), we see that it additionally certifies that the approximate solutions form clusters.
In a similar fashion, we certify that any pair of $(1-\eta)$-satisfying assignments to a random $\ThreeSAT$ instance must have Hamming distance close to 0 or roughly $\frac{n}{2}$, i.e.\ the solutions form clusters where the clusters are roughly $\frac{n}{2}$ apart.
The main ingredient is an efficient algorithm to certify an important structural result of random 3-uniform hypergraphs (\pref{lem:3_hypergraph_structure}), allowing us to reason about the constraints violated in $\pInst$.
In fact, \pref{lem:3_hypergraph_structure} will also be a crucial step in refuting CSPs under global cardinality constraints in \pref{sec:global-cardinality}.  The upshot is that we will be able to certify that any pair of solutions is either $\rho$-close or $\frac{1-\rho}{2}$-far.

The second ingredient is a result in coding theory.
Since the clusters are roughly $\frac{1 \pm \rho}{2}n$ apart in Hamming distance, the number of clusters must be upper bounded by the cardinality of the largest $\rho$-balanced binary error-correcting code.
The best known upper bound is $2^{O(\rho^2\log(1/\rho))n}$ by \cite{MRRW77} (see also \cite{Alo09}), which yields our final result.  Complete details are in \pref{sec:clustering}.

\paragraph{Balance.}  We observe that the idea of hypergraph expansion can be applied to the problem of strongly refuting random CSPs with \textit{global cardinality constraints}.  This problem was first investigated by Kothari, O'Donnell, and Schramm \cite{KOS18}, where they proved that under the refutation threshold $n^{k/2}$, the polynomial-time regime of the Sum-of-Squares hierarchy cannot refute the instance even with the global cardinality constraint $\sum_{i=1}^n x_i = B$ for any integer $B\in [-O(\sqrt{n}), O(\sqrt{n})]$ (here we assume $x\in \{\pm1\}^n$).
On the other hand, they proved once that $|B| > n^{3/4}$, Sum-of-Squares could indeed refute a random $\kXOR$ instance up to a factor of $\sqrt{n}$ under the refutation threshold.

We say an assignment $x$ is $\rho$-biased if $\frac{1}{n}\left|\sum_{i\in[n]}x_i\right|\ge\rho$.  We give a strong refutation algorithm for random instances of all Boolean CSPs under the constraint that the solution is ``unbalanced''.
\begin{theorem}
    Let $P$ be any $k$-variable predicate and let $\bInst$ be a random CSP instance on $m\coloneqq n^{\frac{k-1}{2}+\beta}$ clauses where $\beta > 0$.  For every constant $\rho>0$, there is an efficient algorithm that certifies that $\bInst$ has no $2\rho$-biased assignment which $(1-O(\rho^k))$-satisfies $\bInst$ as a $P$-CSP instance.
\end{theorem}

\begin{remark}
    Compared to \cite{KOS18}, our result is a strong refutation algorithm for all CSPs, whereas their algorithm is specific for $\kXOR$ and only a weak refutation (refuting only exactly satisfying assignments).
    For $k=3$ (\pref{thm:3SAT_global_card}), we match their cardinality constraint requirement (see \pref{rem:comparison_to_KOS}).
    However, for $k\geq 4$ (\pref{cor:main-kCSP-global-card}), we require a slightly stronger cardinality assumption.
\end{remark}

The formal statements and proofs are detailed in \pref{sec:global-cardinality}.  Akin to the case for counting solutions, we employ the reduction of every $\kCSP$ to $\kSAT$ and the $\kXOR$ principle to reduce the problem to strongly refuting $\kXOR$ under global cardinality constraints.

The first main insight is that given a graph $G$ which is a sufficiently good spectral expander, we can efficiently certify that any $\TwoXOR$ instance on $G$, where the number of positive constraints is roughly equal to the number of negative constraints, has no unbalanced approximately satisfying assignments.  The proof of this is based on using the expander mixing lemma to show that any imbalanced assignment $x$ must satisfy $x_ux_v=+1$ for $\gg\frac{1}{2}$ of the edges, which then lets us lower bound the number of negative constraints that are violated.

Then given a random $\kXOR$ instance $\bInst$, we pick some set of $\rho n$ vertices $S$ and consider all clauses with exactly $k-2$ vertices in $S$ and $2$ variables outside $S$.  If we place an edge between the two variables outside $S$ for every clause, we get some random graph $\bG$.  Now consider any assignment $y$ to the variables in $S$.  For this chosen set of clauses to be (nearly) satisfied, the assignment to variables outside $S$ must nearly satisfy the induced $\TwoXOR$ instance on the graph $\bG$ whose signings are determined by $y$.  The second insight is that we can efficiently certify that for any assignment $y$ the induced $\TwoXOR$ instance has a roughly equal number of positive and negative constraints.  This is possible since the quantity $\#\text{positive constraints}(y)-\#\text{negative constraints}(y)$ is the objective value of a particular random $(k-2)\XOR$ instance on assignment $y$, which we can certify tight bounds on using the algorithm of \cite{AOW15}.

\paragraph{Certified counting for subspace problems.}
So far, we have developed certification algorithms for CSPs mainly based on analyses of random hypergraphs.
For other inherently different problems such as counting solutions to the SK model, we turn to a different technique.
Our main insight is that for several problems, the approximate solutions must lie close to a small-dimensional linear subspace.
Thus, we can reduce the problem to counting the number of (Boolean) vectors close to a subspace.
We name this technique \textit{dimension-based count certification} since the algorithms and their guarantees only depend on the dimension of the subspace.

\begin{theorem} \label{thm:subspace_count_main_thm}
    Let $V$ be a linear subspace of dimension $\alpha n$ in $\R^n$.
    For any $\eps\in (0,1/4)$, the number of Boolean vectors in $\cube$ that are $\eps$ away from $V$ is upper bounded by $2^{(H_2(4\eps^2) + \alpha \log \frac{3}{\eps})n}$.
\end{theorem}

We note that the upper bound is almost tight (see \pref{rem:subspace_count_tight} for more details).

We now give a brief overview of the proof of \pref{thm:subspace_count_main_thm}.
First, we upper bound the maximum number of (normalized) Boolean vectors that can lie within any $\eps'$-ball.
Secondly, we take an $\eps$-net of the unit ball in the subspace $V$ (i.e.\ $B_1(0) \cap V$).
We simply multiply the two quantities to get the upper bound, which only depends on the dimension of $V$.

Next, we apply this technique to two problems: the Sherrington-Kirkpatrick model and the independent sets in random $d$-regular graphs.

\paragraph{Sherrington-Kirkpatrick (SK).}
Given $\G$ sampled from $\GOE(n)$, the SK problem is to compute
\begin{equation*}
    \OPT(\G) = \max_{x\in \{\pm1\}^n} x^\top \G x.
\end{equation*}
This problem can also be interpreted as finding the largest cut in a Gaussian-weighted graph.
The SK model arises from the spin-glass model studied in statistical physics \cite{SK75}.
Talagrand \cite{Tal06} famously proved that $\OPT(\G)$ concentrates around $2\mathsf{P}^* n^{3/2} \approx 1.526 n^{3/2}$, where $\mathsf{P}^*$ is the \textit{Parisi constant}, first predicted by Parisi \cite{Par79,Par80}.

Recently, the problem of certifying an upper bound for $\OPT(\G)$ has received wide attention.
A natural algorithm is the \textit{spectral refutation}:
$\OPT(\G) \leq n\cdot \lambda_{\max}(\G)$.
Since $\lambda_{\max}(\G)$ concentrates around $2\sqrt{n}$, the algorithm certifies that $\OPT(\G) \leq (2+o(1)) n^{3/2}$, which we call the \textit{spectral bound}.
Clearly, there is a gap between the spectral bound and the true value, and it is natural to ask whether there is an algorithm that beats the spectral bound.
Surprisingly, building on works by \cite{MS16,MRX20,KB20}, Ghosh et.\ al.\ \cite{GJJ20} showed that even the powerful Sum-of-Squares hierarchy cannot certify a bound better than $(2-o(1))n^{3/2}$ in subexponential time.
We also mention an intriguing work by Montanari \cite{Mon19} where he gave an efficient algorithm for the \emph{search problem} --- to find a solution with objective value close to $\OPT(\G)$ with high probability (assuming a widely-believed conjecture from statistical physics).
However, we emphasize that his algorithm is not a certification algorithm (recall \pref{def:certification_algorithm}).

In the spirit of this work, a natural question is to certify an upper bound on the number of assignments $x\in \{\pm 1\}^n$ such that $x^\top \G x \geq 2(1-\eta)n^{3/2}$ for some $\eta >0$.

\begin{theorem}\label{thm:SK_main_thm}
    Let $\G \sim \GOE(n)$.
    Given $\eta \in (0,\eta_0)$ for some universal constant $\eta_0$,
    there is an algorithm certifying that at most $2^{O(\eta^{3/5}\log\frac{1}{\eta})n}$ assignments $x\in \{\pm1\}^n$ satisfy $x^\top \G x \geq 2(1-\eta)n^{3/2}$.
\end{theorem}

Our proof first looks at the eigenvalue distribution of $\G$, which follows the \textit{semicircle law} (\pref{thm:semicircle_law}).
This shows that any $x$ that achieves close to the spectral bound must be close to the top eigenspace of $\G$ (of dimension determined by the semicircle law).
Then, we directly apply \pref{thm:subspace_count_main_thm}.
See \pref{sec:SK} for complete details.

\paragraph{Independent sets in $d$-regular graphs.}
The largest independent set size (the \textit{independence number}) in a random $d$-regular graph has been studied extensively.
It is well-known that with high probability, the independence number is $\leq \frac{2n\log d}{d}$ for a sufficiently large constant $d$ (cf.\ \cite{Bol81,Wor99}).
The current best known \textit{certifiable} upper bound is via the smallest eigenvalue of the adjacency matrix (often referred to as Hoffman's bound, cf.\ \cite{FO05,BH11}):
Let $\bA$ be the adjacency matrix, and let $\lambda \coloneqq -\lambda_{\min}(\bA)$. Then, $|S| \leq \frac{\lambda}{d+\lambda}n$ for all independent sets $S$.
We give a proof for completeness in \pref{sec:independent_set}.

It is also well-established that $\lambda \leq 2\sqrt{d-1} + o(1)$ with high probability (see \pref{thm:eigenvalue_regular_graph}).
Thus, we can certify that the independence number is at most $\Cindset n$ where $\Cindset \coloneqq \frac{2\sqrt{d-1}}{d+2\sqrt{d-1}}$.

The natural question for us is to certify an upper bound on the number of independent sets larger than $\Cindset(1-\eta)n$ for some $\eta > 0$.

\begin{theorem} \label{thm:ind_set_main_thm}
    For a random $d$-regular graph on $n$ vertices,
    given $\eta \in (0,\eta_0)$ for some universal constant $\eta_0$,
    there is an algorithm certifying that there are at most $2^{O(\eta^{3/5}\log\frac{1}{\eta})n}$ independent sets of size $C_d(1-\eta)n$.
\end{theorem}

The proof is very similar to the SK model.
We first map each independent set $S$ to a vector $y_S\in \R^n$ such that if $S$ is large, then $y_S$ is close to the bottom eigenspace of $\bA$.
Then, using a variant of \pref{thm:subspace_count_main_thm}, we upper bound the number of such vectors that are close to the eigenspace.
We carry out the proof in full detail in \pref{sec:independent_set}.

\paragraph{Optimality for counting $k$CSP solutions.}  Finally, we give evidence suggesting that our algorithmic upper bounds are close to optimal.
Our hardness results are built on the hypothesis that there is no efficient \textit{strong} refutation algorithm for random $\kXOR$ under the refutation threshold (in the regime $n^{\eps} \ll \Delta \ll n^{k/2-1}$).
Although no $\NP$-hardness results are known, this hypothesis is widely believed to be true.
In particular, the problem was shown to be hard for the Sum-of-Squares semidefinite programming hierarchy \cite{Gri01,Sch08,KMOW17}, which is known to capture most algorithmic techniques for average-case problems.
Thus, improving our results would imply a significant breakthrough.

We show that assuming this hypothesis is true, then we cannot certify an upper bound on the number of $(1-\eta)$-satisfying assignments better than $\exp(O(\eta n))$.

\begin{theorem} \label{thm:hardness_main_thm}
    If there is an efficient algorithm that with high probability can certify a bound of $\exp(\frac{\eta n}{10k})$ on the number of $(1-\eta)$-satisfying assignments to $\bInst$, then there is an efficient algorithm that with high probability can certify that $\bInst$ has no $(1-\eta/2)$-satisfying assignments.
\end{theorem}

This shows that the term $\exp(\wt{O}(\eta n))$ in \pref{thm:kCSP_main_thm} and \pref{thm:kxor_main_thm} is tight up to log factors (see \pref{rem:kXOR_count_tight}).
Our proof is simple: given a $(1-\eta/2)$-satisfying assignment and a small set $S$, we can flip the assignments to $S$ arbitrarily and still be $(1-\eta)$-satisfying.  Hence the number of $(1-\eta)$-satisfying assignments is at least $2^{|S|}$.
Thus, an upper bound better than this would imply that there is no $(1-\eta/2)$-satisfying assignments.
See \pref{sec:kXOR_hardness} for complete details.

Surprisingly, the optimality of \pref{thm:kCSP_main_thm} suggests that there is a phase transition for certifiable counting at the refutation threshold.
For concreteness, take random $\kSAT$ for example,

\begin{remark}
    At $m = \wt{\Omega}(n^{k/2})$, there is a strong refutation algorithm \cite{AOW15} which certifies that no $(1-\eta)$-satisfying assignment exists (even for constant $\eta < 1/2$).
    However, at $m = n^{k/2-\eps}$ and take $\eta = n^{-\frac{1}{4}+ \frac{\eps}{2}}$, we can at best certify that the number of $(1-\eta)$-satisfying assignments is at most $\exp(O(n^{\frac{3}{4}+\frac{\eps}{2}}))$.
    See also \pref{fig:plot-of-results} for illustration.
\end{remark}

\paragraph{Optimality for counting independent sets.}
We also show barriers to improving \pref{thm:ind_set_main_thm}, which can be viewed as a weak hardness evidence.
Specifically, we show that improving the upper bound of \pref{thm:ind_set_main_thm} to $\exp\left(O(\eta \log(1/\eta)n)\right)$ would imply beating Hoffman's bound by a factor of $1-\eta/2$ (for any small constant $\eta$), which would be an interesting algorithmic breakthrough.

\begin{theorem}
    Let $\bG$ be a random $d$-regular graph.
    Given constant $\eta \in (0, 1/2)$, if there is an efficient algorithm that with high probability certifies a bound of $\exp\left(\frac{\Cindset}{4} \eta \log(1/\eta)n \right)$ on the number of independent sets of size $\Cindset (1-\eta) n$, then there is an algorithm that with high probability certifies that $\bG$ has no independent set of size $(1-\eta/2)\Cindset n$.
\end{theorem}

The proof is a simple observation that for any independent set $S$, all subsets of $S$ are also independent sets.
Thus, if $S$ is of size $(1-\eta/2)\Cindset n$, then we can lower bound the number of subsets of size $(1-\eta)\Cindset n$.
We give a short proof in \pref{sec:ind_set_hardness}.
We note the interesting gap between $\eta^{3/5}$ and $\eta$ in the exponent of the upper and lower bounds respectively, and we conjecture that there may be an algorithm matching the lower bound.

\subsection{Context and related work} \label{sec:related_work}

\paragraph{Information-computation gaps in CSPs.}  This work is very closely related to the line of work on information-computation gaps.  In the context of certification in random CSPs, the most well-understood information-gaps are in that of refutation of random CSPs.  Feige's random $\ThreeSAT$ hypothesis was one of the earliest conjectured gaps.  As discussed earlier, while unsatisfiability for random $\ThreeSAT$ set in at constant density, it was conjectured by Feige that certifying this was hard at all constant densities.  Further, integrality gaps for the Sum-of-Squares hierarchy of \cite{Gri01,Sch08} seem to point to hardness up to density $\sqrt{n}$.
The wide information-computation gap is a main motivation for us to understand what an efficient algorithm can certify about the landscape of solutions in the regime between the satisfiability threshold and the refutation threshold.
We refer the reader to the introduction of \cite{AOW15} for a comprehensive treatment of the literature on information-computation gaps for refuting random CSPs prior to their work, CSPs more broadly, as well as connections to other areas of theoretical computer science.

The situation for general constraint satisfaction problems beyond $\XOR$ and $\SAT$ was considered in the work of \cite{AOW15}, which gave algorithms to refute all CSPs at density $n^{t/2-1}$ where $t$ is the smallest integer such that there is no $t$-wise uniform distribution supported on the predicate's satisfying assignments.  Then somewhat surprisingly, the work of \cite{RRS17} gave algorithms for refuting random CSPs between constant density and the $n^{t/2-1}$ threshold from \cite{AOW15}, whose running time smoothly interpolated between exponential time at constant density to polynomial time at the \cite{AOW15} threshold, with a (steadily improving) subexponential running time in the intermediate regime.  The algorithms of \cite{AOW15,RRS17} are spectral, and can be captured within the Sum-of-Squares hierarchy.  Finally the work of \cite{KMOW17} (presaged by \cite{BCK15}) established that the algorithm of \cite{RRS17} was tight for Sum-of-Squares in all regimes, thereby nailing a characterization for the exact gaps (up to logarithmic factors) for all random CSPs.

\paragraph{Solution geometry in random CSPs.}  One of the earlier predictions using nonrigorous physics techniques was the location of the $\ThreeSAT$ satisfiability threshold in the works of \cite{MZ02,MPZ02}.   In particular, they conjectured that there is a sharp threshold at a constant $\alpha_{\mathrm{SAT}}\approx 4.267$.  These works put forth the ``1-step replica symmetry breaking hypothesis'' (a conjectured property of the solution space in random $\kSAT$; we refer the reader to the introduction of \cite{DSS15} for a description), which was the starting point for several subsequent works.   These techniques were used to precisely predict the $\kSAT$ satisfiability threshold for all values of $k$ \cite{MMZ06}, proved for large $k$ in a line of work culminating in \cite{DSS15} and building on \cite{AM02,AP03,CP13,CP16}.

Eventually, the works of \cite{KMRSZ07,MRS08} predicted that besides the satisfiability threshold, random $\kSAT$ goes through other phase transitions too, and gave conjectures for their locations.  A notable one connected to this work is the \emph{clustering threshold}, for which there has been rigorous evidence given in the works of \cite{MMZ05,AR06,AC08}.  Above the clustering threshold, the solution space is predicted to break into exponentially many exponential-sized clusters far away from each other in Hamming distance.  More precisely, there is some function $\Sigma$ for which there are $\exp(\Sigma(s,\Delta)n)$ clusters of size approximately $\exp(sn)$ each.  In particular, this leads to the prediction that the number of solutions at density $\Delta$ is roughly $\max_{s}\{\exp((s+\Sigma(s,\Delta))n)\}$.  Another phase transition of interest is the \emph{condensation threshold}, where the number of clusters of solutions drops to a constant.

\paragraph{Approximate Counting for CSPs.}
Approximate counting of solutions in CSPs has attracted much attention in recent years.  There have been numerous positive algorithmic results for approximately counting solutions in (i) sparse CSPs in the worst case, (ii) sparse random CSPs well under the satisfiability threshold.  The takeaway here is that even though the problems we consider get harder as we approach the satisfiability threshold, if one goes well under the threshold the algorithmic problems once again become tractable.

One exciting line of research for worst-case CSPs is the problem of approximately counting satisfying assignments of a $\kSAT$ formula under conditions similar to those of the \Lovasz Local Lemma (LLL) \cite{ER73}.
A direct application of the LLL shows that if the maximum degree $D$ of the \textit{dependency graph} is $\leq 2^k/e$, then the formula is satisfiable.
Building on works of \cite{Moi19,FGYZ20,FHY20,JPV20}, Jain, Pham, and Vuong \cite{JPV21} recently showed that there is an algorithm for approximate counting well under the LLL thresholds, i.e. when $D \lesssim 2^{k/5.741}$ (hiding factors polynomial in $k$), using techniques similar to an algorithmic version of the LLL.
Further, the algorithms of \cite{Moi19,JPV20} are deterministic, which may suggest their techniques are amenable to obtaining certifiable counts.
However, it was shown that the problem of approximately counting solutions to a $\kSAT$ formula is $\NP$-hard when $D \gtrsim 2^{k/2}$ by \cite{BGG19}, well in the sparse regime, which suggests a hard phase between the highly sparse setting and the dense setting we are concerned with.

For random $\kSAT$, the exact satisfiability threshold that was established by Ding, Sly, and Sun \cite{DSS15} takes on value $\alpha_{\mathrm{SAT}} = 2^k \ln 2 - \frac{1}{2}(1+\ln2) + o_k(1)$.  And similarly, well below the satisfiability threshold, Galanis, Goldberg, Guo, and Yang \cite{GGGY19} adapted Moitra's techniques \cite{Moi19} to the random setting and developed a polynomial-time algorithm when the density $\Delta \leq 2^{k/301}$ and $k$ sufficiently large.

Closely related to the counting problem is approximating the partition function of random $\kSAT$, for which there have also been positive algorithmic results.
Specifically, given a random $\kSAT$ instance $\bInst$, the partition function is defined as $Z(\bInst,\beta) \coloneqq \sum_{\sigma} e^{-\beta H(\sigma)}$, where $H(\sigma)$ is the number of unsatisfied clauses under assignment $\sigma$.
The partition function can be viewed as a weighted (or ``permissive'') version of the counting problem.
Montanari and Shah \cite{MS06} first showed that the Belief Propagation algorithm approximately computes the partition function at $\Delta \sim \frac{2\log k}{k}$;
their analysis is based on correlation decay (or the \textit{Gibbs uniqueness property}).
Recently, \cite{COMR20} further showed that Belief Propagation succeeds as long as the random $\kSAT$ model satisfies a \textit{replica symmetry} condition, conjectured to hold up to $\Delta \sim 2^k \ln k / k$.
See also the works of \cite{KMRSZ07,Pan13,CO17} for further details of this matter.

\paragraph{Counting independent sets and related problems.}
Another counting problem that has been the subject of active study is that of counting independent sets, especially in the statistical physics community.
For a graph $G$ with maximum degree $d$, let $\mathsf{IS}(G)$ be the set of independent sets in $G$.
The task is to estimate the \textit{independence polynomial} $Z_G(\lambda) = \sum_{I\in \mathsf{IS}(G)} \lambda^{|I|}$, also known as the partition function of the \textit{hard-core model} with \textit{fugacity} $\lambda$ in the physics literature.
Earlier works by \cite{DG00,Vig01} developed randomized algorithms based on \textit{Glauber dynamics} to estimate $Z_G(\lambda)$ when $\lambda \leq \frac{2}{d-2}$.
In a major breakthrough, Weitz \cite{Wei06} showed a deterministic algorithm, based on correlation decay, that approximates $Z_G(\lambda)$ when $0 \leq \lambda < \lambda_c$, where $\lambda_c \coloneqq \frac{(d-1)^{d-1}}{(d-2)^{d}}$.
Sly and Sun \cite{SS12} later proved that this is tight: no efficient approximate algorithm for $Z_G(\lambda)$ exists for $\lambda > \lambda_c$ unless $\NP = \RP$.

Recently, Barvinok initiated a line of research on estimating partition functions using the \textit{interpolation method}
(see Barvinok's recent book \cite{Bar16}).
The main idea is to estimate the low-order Taylor approximation of $\log Z_G(\lambda)$ provided that the polynomial $Z_G(\lambda)$ does not vanish in some region in $\C$.
This approach led to deterministic algorithms that match Weitz's result and work even for negative or complex $\lambda$'s \cite{PR17,PR19}.  These polynomial-based approaches were also used to obtain deterministic algorithms for counting colorings in bounded degree graphs \cite{LSS19a}, estimating the Ising model partition function \cite{LSS19b}, and algorithms for a counting version of the Unique Games problem \cite{CDKPR19}.

There has also been works on worst-case upper bounds of $Z_G(\lambda)$ for $d$-regular graphs.
Zhao proved that for any $d$-regular graph $G$ and any $\lambda\geq 0$, $Z_G(\lambda) \leq (2(1+\lambda)^d - 1)^{n/2d}$ \cite{Zha10}.
In particular, setting $\lambda =1$, this shows that the total number of independent sets is bounded by $(2^{d+1}-1)^{n/2d}$, settling a conjecture by Alon \cite{Alo91} and Kahn \cite{Kah01}.

\paragraph{Certifying bounds on partition functions and free energy.}  A recent line of work \cite{Ris16,RL16,JKR19} is focused on an approach based on a convex programming relaxation of entropy to certify upper bounds on the \emph{free energy} of the Ising model (weighted $\TwoXOR$), both in the worst case and in the average case.  While on the surface level, these approaches differ significantly from ours, an interesting direction is to investigate if these entropy-based convex programming relaxations can achieve our algorithmic results.






\subsection{Table of results}
We include a table to have a succinct snapshot of our results.

\newcommand{\ColorCap}{\color{MidnightBlue}}

\renewcommand{\arraystretch}{1.5}
\begin{table}[!ht]
    \centering
    \begin{tabular}{|c|c|c|c|c|}
    \hline
                              & \textbf{Problem}         & \textbf{Theorem}                           & \textbf{Upper bound}                                                                                                                                                                                                             & \textbf{Randomness}                                                      \\ \hline \hline
    \multirow{5}{*}{\textbf{Counts}}   & $\TwoXOR$       & $\ref{thm:2xor_upper_bound}$      & $\exp(\wt{O}(\eta n))$                                                                                                                                                                                                  & Hypergraph                                                      \\ \cline{2-5} 
                              & $\kXOR$         & $\ref{thm:kxor_upper_bound}$      & $\exp(\wt{O}(\eta n))\cdot\exp\left(O\left(\frac{n^{1+\eps}}{\Delta^{1/(k-2)}}\right)\right)$                                                                                                                           & Hypergraph                                                      \\ \cline{2-5} 
                              & $k$CSP          & $\ref{cor:kCSP-to-kSAT}$          & \begin{tabular}[c]{@{}c@{}}$\exp(\wt{O}(\eta n)) \cdot \exp\left(O\left(\frac{n^{1+\eps}}{\Delta^{1/(k-2)}}\right)\right)$\\ $\cdot \exp\left(\wt{O}\left(\sqrt{\frac{n^{(k+1)/2}}{\Delta}}\right)\right)$\end{tabular} & \begin{tabular}[c]{@{}c@{}}Hypergraph\\ + signings\end{tabular} \\ \cline{2-5} 
                              & SK model        & $\ref{thm:SK_upper_bound}$        & $\exp(O(\eta^{3/5}\log\frac{1}{\eta})n)$                                                                                                                                                                                & $\G \sim \GOE(n)$                                               \\ \cline{2-5} 
                              & Independent set & $\ref{thm:ind_set_upper_bound}$   & $\exp(O(\eta^{3/5}\log\frac{1}{\eta})n)$                                                                                                                                                                                & $d$-regular graph                                               \\ \hline
    \multirow{2}{*}{\textbf{Clusters}} & $\ThreeXOR$     & $\ref{thm:3XOR_clusters}$         & \begin{tabular}[c]{@{}c@{}}$\exp(O(\theta^2 \log(1/\theta))n)$,\\ for $\theta = \max(2\eta, \Delta^{-\frac{1}{2}}\log n)$\end{tabular}                                                                                   & Hypergraph                                                      \\ \cline{2-5} 
                              & 3CSP            & $\ref{cor:3SAT_clusters}$         & \begin{tabular}[c]{@{}c@{}}$\exp(O(\theta^2 \log(1/\theta))n)$,\\ for $\theta = 8\eta + \wt{O}(\Delta^{-1/2})$\end{tabular}                                                                                              & \begin{tabular}[c]{@{}c@{}}Hypergraph\\ + signings\end{tabular} \\ \hline
    \multirow{2}{*}{\textbf{Balance}}  & 3CSP            & $\ref{thm:3SAT_global_card}$      & \begin{tabular}[c]{@{}c@{}}bias $\leq \rho$ for $\rho \gg \sqrt{\frac{\log n}{\Delta}}$,\\ $\eta = \rho/16$\end{tabular}                                                                                                & \begin{tabular}[c]{@{}c@{}}Hypergraph\\ + signings\end{tabular} \\ \cline{2-5} 
                              & $k$CSP          & $\ref{cor:main-kCSP-global-card}$ & \begin{tabular}[c]{@{}c@{}}bias $\leq \rho$ for any constant $\rho > 0$,\\ $\eta = \rho^k/2^{k+1}$\end{tabular}                                                                                                         & \begin{tabular}[c]{@{}c@{}}Hypergraph\\ + signings\end{tabular} \\ \hline
\end{tabular}
\caption{A summary of our results. \newline
        (1)\quad \textbf{$k$CSP counts}: given $\bInst\sim\HDist{k}{m}{n}$ where $m = \Delta n$, we upper bound the number of $(1-\eta)$-satisfying assignments. \newline
        (2)\quad \textbf{SK counts}: given $\G\sim\GOE(n)$, we upper bound the number of $x\in \{\pm 1\}^n$ such that $x^\top \G x \geq 2(1-\eta)n^{3/2}$. \newline
        (3)\quad \textbf{Independent set counts}: given a random $d$-regular graph for constant $d\geq 3$, we upper bound the number of independent sets of size $\geq \Cindset n(1-\eta)$. \newline
        (4)\quad  \textbf{3CSP clusters}: given $\bInst\sim\HDist{3}{m}{n}$ where $m = \Delta n$, we upper bound the number of diameter-$(\theta n)$ clusters of $(1-\eta)$-satisfying assignments. \newline
        (5)\quad \textbf{Balance}: given $\bInst\sim\HDist{k}{m}{n}$ where $m = \Delta n$, we certify that any $(1-\eta)$-satisfying assignment must have \textit{bias} $\frac{1}{n}\left|\sum_{i\in[n]} x_i\right| \leq \rho$.
    }
\label{table:results}
\end{table}

\subsection{Open directions}
In this section we suggest a couple of avenues for further investigation on the themes related to this work.

\paragraph{Worst-case complexity of certified counting.}  In this work, we deal mostly with random CSPs.  Here we present a worst-case version of the problem, specialized to $\ThreeSAT$.  A classic result due to \cite{Has01} is that it is $\NP$-hard to distinguish between a $(7/8+\eps)$-satisfiable $\ThreeSAT$ formula from a fully satisfiable $\ThreeSAT$ formula.  However, it is unclear what the complexity of a version of this question is when there is a stronger promise on the satisfiable $\ThreeSAT$ formula.
\begin{question}
    Consider the following algorithmic task:
    \begin{displayquote}
        Given a $\ThreeSAT$ formula $\calI$ under the promise that it is either $(7/8+\eps)$-satisfiable, or has at least $T$ fully satisfying assignments, decide which of the two categories $\calI$ falls into.
    \end{displayquote}
    What is the complexity of the above problem?
\end{question}
We remark that this problem is similar to counting-$\ThreeSAT$, but subtly different.

\paragraph{Certifying optimal bounds on number of exactly satisfying $\kSAT$ solutions.}  In the context of $\kSAT$, while our algorithms can certify subexponential bounds for both exactly satisfying assignments and approximately satisfying assignments, the matching evidence of hardness is only for the approximate version of the problem.  Thus, it is still possible that there is an algorithm to certify an even tighter bound than ours for the problem of counting exactly satisfying assignments to a random $\kSAT$ formula.  This motivates the following question:
\begin{question}
    What is the tightest bound an efficient algorithm can certify on the number of solutions to a random $\kSAT$ instance?
\end{question}
We conjecture that the algorithms presented in this paper are indeed optimal.  An approach to providing hardness evidence for this is to construct a hard planted distribution, and prove it is hard within the \emph{low-degree likelihood ratio} framework of \cite{HS17}.
We outline a possible approach in \pref{sec:kSAT-hardness-approach} to construct a planted distribution for readers interested in this problem.

\paragraph{Properties of arbitrary CSP instances on random hypergraphs.}  In the context of approximate $\kXOR$, our certification algorithms for solution counts and cluster counts depend only on the hypergraph structure and not the random negations. Hence, they also prove nontrivial statements about the solution space of any $\XOR$ instance on a random hypergraph, which are potentially useful in the context of quiet planting or semi-random models of CSPs.  However, our certification algorithms for other CSPs, such as $\kSAT$, heavily make use of the random signings in the reduction to $\kXOR$.
\begin{question}
    Can all the results related to certifying bounds on number of solutions/clusters in this work for random $\kSAT$ instances be generalized to arbitrary $\kSAT$ instances on random hypergraphs?
\end{question}

\section{Preliminaries}





\subsection{Graph theory}

Given a graph $G$, we use $V(G)$ to denote its vertex set, $E(G)$ to denote its edge set, and $\deg_G(u)$ to denote the degree of a vertex $u$.
For $S\subseteq V(G)$ and $T\subseteq V(G)$, we use $e(S,T)$ to denote the number of tuples $(u,v)$ such that $u\in S$, $v\in T$ and $\{u,v\}\in E(G)$.  

We will be concerned with its \emph{normalized Laplacian matrix}, denoted $L_G$, defined as:
\[
	L_G \coloneqq \Id - D_G^{-1/2}A_G D_G^{-1/2},
\]
where $A_G$ is the adjacency matrix of $G$ and $D_G$ is the diagonal matrix of vertex degrees in $G$.  Since $L_G$ is a self-adjoint matrix, it has $n$ real eigenvalues, which we sort in increasing order and denote as:
\[
	0 = \lambda_1(G) \le \lambda_2(G) \le \dots \le \lambda_n(G).
\]
Of particular interest to us is $\lambda_2(G)$, which we call the \emph{spectral gap}.

A combinatorial quantity we will be concerned with is the \emph{conductance} of $G$.  For a subset $S\subseteq V$, we define the \emph{volume} of $S$ as $\vol(S)\coloneqq \sum_{u\in S}\deg(u)$ and let $\phi_G(S)\coloneqq \frac{e(S,\ol{S})}{\vol(S)}$.  The conductance of $G$ is then defined as:
\[
	\phi_G \coloneqq \min_{\substack{S\subseteq V(G) \\ \vol(S)\le \vol(V)/2 }} \phi_G(S).
\]
The well-known Cheeger's inequality on graphs, first proved in \cite{AM85}, relates the conductance and the spectral gap.  We refer the reader to \cite{T17} for a good exposition of the proof.
\begin{theorem}[Cheeger's inequality]
\label{thm:cheeger}
	For any graph $G$,
	\[
		\frac{\lambda_2(G)}{2} \le \phi_G \le \sqrt{2\lambda_2(G)}.
	\]
\end{theorem}

It is well-known that dense \ER random graphs have large spectral gaps (cf. \cite{CO07,HKP19}).
\begin{theorem}[{\cite[Theorem 1.1]{HKP19}}]
\label{thm:spectral_gap}
    Let $\bG$ be an \ER random graph with $p = \omega\left(\frac{\log n}{n}\right)$, and let $d = p(n-1)$ denote the average degree.
    Then, there is a constant $C$ such that
    \[
        \lambda_2(\bG) \geq 1- \frac{C}{\sqrt{d}}
    \]
    with probability at least $1 - Cn \exp(-d) - C\exp(-d^{1/4}\log n)$.
\end{theorem}

Closely related to Cheeger's inequality is the \textit{expander mixing lemma}, which roughly states that the edges of an expander graph are \textit{well distributed}.
Here, we consider the adjacency matrix $A$ and the ``de-meaned'' matrix $\meanA \coloneqq A - \frac{d}{n}J$, where $J$ is the all-ones matrix.
We include a short proof for completeness.

\begin{theorem}[Expander Mixing Lemma \cite{AC88}]
\label{thm:expander_mixing_lemma}
	Let $G$ be a graph with $n$ vertices and average degree $d$, and let $\meanA$ be the de-meaned adjacency matrix.
	Then, for any $S,T \subseteq V(G)$,
	\begin{equation*}
		\left|e(S,T) - \frac{d}{n} |S|\cdot |T| \right| \leq \|\meanA\| \sqrt{|S| \cdot |T|}.
	\end{equation*}
\end{theorem}

\begin{proof}
    Let $1_S, 1_T\in \zo^{n}$ be the indicator vectors for subsets $S,T$.
    Clearly, we have $1_S^\top A 1_T = e(S,T)$ and $1_S^\top J 1_T = |S|\cdot |T|$.
    Moreover, $\|1_S\|_2 = \sqrt{|S|}$ and $\|1_T\|_2 = \sqrt{|T|}$.
    Thus,
    \begin{equation*}
    \begin{gathered}
        e(S,T) - \frac{d}{n}|S|\cdot |T| = 1_S^\top \left(A - \frac{d}{n}J\right) 1_T \\
        \Rightarrow \left| e(S,T) - \frac{d}{n}|S|\cdot |T| \right| \leq \left\| \meanA \right\| \cdot \|1_S\|_2 \cdot \|1_T\|_2.
    \end{gathered}
    \qedhere
    \end{equation*}
\end{proof}

\subsection{Fourier analysis of Boolean functions}
We refer the reader to \cite{O14} for an elaborate treatment of the subject.  The functions $\{\prod_{i\in \T}x_i\}_{\T\subseteq[k]}$ form an orthogonal basis for the space of functions from $\{\pm1\}^k$ to $\R$, and hence any function $f:\{\pm1\}^k\to\R$ can be expressed as a multilinear polynomial:
\[
	f(x) = \sum_{\T\subseteq[n]} \wh{f}(\T)\prod_{i\in \T}x_i.
\]
Further, the coefficients $\wh{f}(\T)$, which are called \emph{Fourier coefficients}, can be obtained via the formula:
\[
	\wh{f}(\T) = \E_{\bz\sim\{\pm 1\}^k}\left[f(\bz)\prod_{i\in \T}\bz_i\right].
\]
A key property, called \emph{Plancherel's theorem}, is the following:
\begin{fact}	\label{fact:plancherel}
	Let $f,g$ be functions from $\{\pm 1\}^k$ to $\R$.  Then:
	\[
		\langle f, g\rangle \coloneqq \E_{\bz\sim\{\pm 1\}^k} f(\bz)g(\bz) = \sum_{\T\subseteq[k]}\wh{f}(\T)\wh{g}(\T).
	\]
\end{fact}

Given a probability distribution $\calD$ on $\{\pm 1\}^k$, following the notation of \cite{AOW15}, we use $D$ to denote the function equal to its probability density function multiplied by $2^k$.  This leads to the notational convenience that for any function $f:\{\pm1\}^k\to\R$,
\[
	\E_{\bz\sim\calD}f(\bz) = \E_{\bz\sim\{\pm 1\}^k} f(\bz) D(\bz).
\]

\subsection{Random hypergraphs and CSPs}	\label{sec:random-hypergraphs}
In this work, we will deal with random CSPs defined as signed $k$-uniform hypergraphs.
\begin{definition}[{Signed $k$-uniform hypergraph}]
	A \emph{signed $k$-uniform hypergraph} $\calI$ on universe $[n]$ is a collection of pairs $\{(c,\cls)\}$ where for each $(c,\cls)$, $c\in\{\pm1\}^k$ and $\cls$ is in $[n]^k$.  Each pair $(c,\cls)\in\calI$ is also called a \emph{hyperedge} in $\calI$.
\end{definition}

The distribution over signed $k$-uniform hypergraphs we work with is:
\begin{definition}  \label{def:signed_hypergraph}
	We use $\HDist{k}{m}{n}$ to denote the distribution on signed $k$-uniform hypergraphs on universe $[n]$ where a sample $\bInst\sim\HDist{k}{m}{n}$ is obtained by independently including each of the $2^kn^k$ potential hyperedges with probability $\frac{m}{2^kn^k}$.
	Moreover, we use $\UnsignedH{k}{m}{n}$ to denote the distribution of (unsigned) $k$-uniform hypergraphs, which is the same as $\HDist{k}{m}{n}$ but with the signs removed.
\end{definition}

Given a tuple $\cls\in[n]^k$ and $x\in\{\pm1\}^n$, we use $x_\cls$ to denote the tuple $(x_{\cls[1]}, \dots, x_{\cls[k]})$.  And for a pair of $k$-tuples $a,b$ we use $a\circ b$ to denote the entrywise product of the tuples $(a_1\cdot b_1, \dots, a_k\cdot b_k)$.

\begin{definition}
	A \emph{constraint satisfaction problem instance} (CSP instance) on variable set $[n]$ is a signed hypergraph $\calI$ along with a function $P:\{\pm1\}^k\to\R$, called a \emph{predicate}.  Given $x\in\{\pm 1\}^n$, we use $P_{\calI}(x)$ to denote the \emph{objective value} on $x$:
	\[
		P_{\calI}(x) \coloneqq \frac{1}{|\calI|} \sum_{(c,\cls)\in\calI} P(c\circ x_\cls).
	\]
\end{definition}

\begin{definition}
	Given a signed $k$-uniform hypergraph $\calI$ on $[n]$ and any $x\in\{\pm 1\}^n$, we define $\calD_{\calI,x}$ to be the distribution on $\{\pm1\}^k$ with density (scaled by $2^k$):
	\[
		D_{\calI,x}(z) \coloneqq \frac{2^k}{|\calI|}\cdot\left|\{(c,\cls)\in\calI | c\circ x_\cls = z \}\right|.
	\]
\end{definition}

Note that $\wh{D_{\bInst,x}}(\varnothing) = \E_{\bz\sim \{\pm1\}^k}[D_{\Inst,x}(\bz)] = 1$ for all $x\in \{\pm 1\}^n$.  A simple but important observation is the following.
\begin{observation}	\label{obs:obj-val}
	If $\calI$ is a signed $k$-uniform hypergraph and $P:\{\pm1\}^k\to\R$ is some predicate, then:
	\begin{align*}
		P_{\calI}(x) &= \frac{1}{|\calI|}\sum_{(c,\cls)\in\calI} P(c\circ x_\cls) \\
		&= \E_{\bz\sim\{\pm1\}^k}\left[P(\bz)D_{\calI,x}(\bz)\right] \\
		&= \sum_{\T\subseteq[k]} \wh{P}(\T) \wh{D_{\calI,x}}(\T).
	\end{align*}
\end{observation}

\begin{definition}[Biased assignment]
	We say $x\in\{\pm1\}^n$ is \emph{$\eta$-biased} if $\frac{1}{n}\left|\sum_{i=1}^n x_i\right| \ge \eta$.
\end{definition}

Next, we focus on $\kXOR$.
Given a pair $(c,\cls)\in \calI$, $P(c\circ x_\cls) = 1$ if and only if $\prod_{i=1}^k x_{\cls[i]} = \prod_{i=1}^k c_i$, thus in the case of $\kXOR$ we will also write $(b,\cls)$ as a constraint (or clause) of $\calI$, where $b = \prod_{i=1}^k c_i \in \{\pm1\}$ is the \textit{signing} of the constraint.  We call a constraint $(b,\cls)$ \emph{positive} if $b=+1$ and \emph{negative} if $b=-1$.  We say an instance $\calI$ is \emph{$p$-positive} if the fraction of its constraints that are positive is at most $p$.
Moreover, we will need the notion of \textit{induced $\XOR$} and \textit{truncated $\XOR$}, which we define below.

\begin{definition}[{Induced $\tXOR$ and truncated $(k-t)\XOR$}]
	\label{def:induced_txor}
    Given a $\kXOR$ instance $\Inst$, an integer $1 \leq t \leq k-1$, a subset of variables $S\subseteq [n]$, and an assignment $\sigma\in \{\pm1 \}^S$, we define the \textit{induced $\tXOR$ instance} $\Inst_{S,\sigma,t}$ on variables $\ol{S}$ as follows:
    for each clause $(b,\cls) \in \Inst$ where all variables in $\cls[1:k-t]$ are in $S$ and all variables in $\cls[k-t+1:k]$ are in $\ol{S}$, add a $\tXOR$ clause $(b', \cls[k-t+1:k])$ where $b' = b \cdot \prod_{i=1}^{k-t}\sigma_i$. 
	Similarly, we define the \emph{truncated $(k-t)\XOR$ instance} $\TruncInst{S}{k-t}$ on variables $S$ as follows: for each clause $(b,\cls)\in\Inst$ where all variables in $\cls[1:k-t]$ are in $S$ and all variables in $\cls[k-t+1:k]$ are in $\ol{S}$, add a $(k-t)\XOR$ clause $(b,\cls[1:k-t])$.

\end{definition}

For example, consider a $4\XOR$ instance and a constraint $x_a x_b x_c x_d = +1$.
Suppose $a, b\in S$, $c, d \notin S$, and $\sigma_a = +1, \sigma_b = -1$.
Then for $t=2$, the induced instance $\Inst_{S,\sigma,2}$ will have a constraint $x_c x_d = -1$, and the truncated instance $\TruncInst{S}{2}$ will have a constraint $x_a x_b = +1$.
Note that the truncated instance does not depend on the assignments to $S$.

\begin{observation} \label{obs:trunc-vs-induced}
	Given a $\kXOR$ instance $\Inst$, subset $S\subseteq [n]$ and $\sigma\in \{\pm1\}^S$, the following are useful relationships between the induced $\tXOR$ instance $\Inst_{S,\sigma,t}$ and the truncated $(k-t)\XOR$ instance $\TruncInst{S}{k-t}$:
	\begin{enumerate}
		\item $\left|\Inst_{S,\sigma,t}\right| = \left|\TruncInst{S}{k-t}\right|$.
		\item $\displaystyle \sum_{(b,U)\in \Inst_{S,\sigma,t}} b = \sum_{(b,U)\in\TruncInst{S}{k-t}} b\prod_{i\in U}\sigma_i$.
	\end{enumerate}	
\end{observation}

\begin{definition}[Induced hypergraph]
	\label{def:induced_hypergraph}
    Given a $\kXOR$ instance $\Inst$, an integer $1 \leq t \leq k-1$, and a subset of variables $S\subseteq [n]$, we define the \textit{induced $t$-uniform hypergraph} $H_{S,t}$ on $n-|S|$ variables as the underlying $t$-uniform hypergraph of the induced $\tXOR$ of $S$.
    Note that the hypergraph only depends on $S,t$ and not the assignment $\sigma$ to $S$.
\end{definition}

For simplicity, we denote $\Inst_{S,\sigma}$ to be $\Inst_{S,\sigma,k-1}$ and $H_{S}$ to be $H_{S,k-1}$.

\begin{definition}[Primal graph]
	\label{def:primal_graph}
	Given a hypergraph $H$, we define the \textit{primal graph} of $H$ on the same vertex set, denoted $\primalH$, as follows: for every hyperedge $e\in H$ and every pair $(u,v)\in e$, add $(u,v)$ to $\primalH$.
	Parallel edges are allowed.
\end{definition}

\subsection{Refuting random CSPs}
We will need the refutation algorithm for random CSPs of \cite{AOW15}.  A crucial notion in \cite{AOW15} is that of approximate $t$-wise uniformity.
\begin{definition}[{$(\eps,t)$-wise uniform}] \label{def:eps_t_wise_uniform}
	A distribution $\calD$ is $(\eps,t)$-wise uniform if for all $\T\subseteq[k]$ such that $0 < |\T|\le t$, $|\wh{D}(\T)|\le\eps$.
\end{definition}

\begin{definition} \label{def:quasirandom}
	We say a signed $k$-uniform hypergraph $\calI$ on $[n]$ is $(\eps,t)$-quasirandom if for every $x\in\{\pm1\}^n$ the distribution $\calD_{\calI,x}$ is $(\eps,t)$-wise uniform.
\end{definition}

Now we are ready to state the key statement we use from \cite{AOW15}.
\begin{theorem}	\label{thm:aow-ref}
	Let $\bInst\sim\HDist{k}{m}{n}$ with $m \ge \alpha n^{t/2}$.
	Then there is an efficient algorithm that with probability $1-o(1)$ certifies that $\bInst$ is $\left(\frac{2^{O(t)}\log^{5/2}n}{\sqrt{\alpha}}, t\right)$-quasirandom.
\end{theorem}

Another statement we use from \cite{AOW15} is their algorithm to refute random polynomials on the hypercube.
\begin{theorem}	\label{thm:aow-poly}
	For $k\ge 2$, let $\{\bw_T\}_{T\in[n]^k}$ be independent centered random variables on $[-1,1]$ such that:
	\begin{align*}
		\Pr[\bw_T \ne 0] \le p \quad \forall T\in [n]^k.
	\end{align*}
	Then there is an efficient algorithm which certifies with high probability:
	\[
		\sum_{T\in[n]^k} \bw_T x^T \le 2^{O(k)} \max\{\sqrt{p} n^{3k/4}, n^{k/2}\} \log^{3/2} n
	\]
	for all $x\in\{\pm1\}^n$.
\end{theorem}

\subsection{Random matrix theory}

We will need the matrix Bernstein inequality for proving spectral norm bounds as stated in \cite[Theorem 6.1.1]{Tro15}.
\begin{theorem}[{Matrix Bernstein inequality (special case of \cite[Theorem 6.1.1]{Tro15})}]	\label{thm:spec-norm-bound}
	Let $\bS_1,\dots,\bS_\ell$ be a collection of independent random symmetric matrices of dimension $d\times d$.  Assume that $\E\bS_i=0$ and $\|\bS_i\|\le 1$ for all $i\in[\ell]$.  Let $\bZ=\sum_{i\in[\ell]}\bS_i$.  Define:
	\[
		v = \left\|\sum_{i\in[\ell]} \E[\bS_i^2] \right\|.
	\]
	Then, for any $t\geq 0$,
	\[
		\Pr[\|\bZ\|\ge t] \le 2d\exp\left(\frac{-t^2}{v+t/3}\right).
	\]
\end{theorem}


We will also need the eigenvalue distribution of random matrices (in \pref{sec:subspace_count}).
We first consider Gaussian matrices.
Let $\bW$ be a random $n\times n$ matrix with independent standard Gaussian entries, and let $\G \coloneqq \frac{1}{\sqrt{2}}(\bW + \bW^\top)$.
We say that $\G$ is sampled from the Gaussian Orthogonal Ensemble, denoted $\GOE(n)$.
We recall the following results in random matrix theory,

\begin{fact}
\label{fact:norm_of_goe}
	The spectral norm $\lambda_{\max}(\G) \leq (2+t) \sqrt{n}$ with probability $1- 2\exp(-nt^2/2)$.
\end{fact}

\begin{theorem}[Semicircle law \cite{Erd11}]
\label{thm:semicircle_law}
	The empirical distribution of eigenvalues of $\G\sim \GOE(n)$ follows a universal pattern, the Wigner semicircle law.
	For any fixed real numbers $a < b$, with probability $1-o(1)$,
	\begin{equation*}
		\frac{1}{n} \left| \left\{i: \frac{1}{\sqrt{n}}\lambda_i(\G) \in [a,b] \right\}\right| = (1\pm o(1)) \int_a^b \rho_{sc}(x)dx, \quad
		\rho_{sc}(x) \coloneqq \frac{1}{2\pi} \sqrt{(4-x^2)_+},
	\end{equation*}
	where $(a)_+ = \max(a, 0)$.
\end{theorem}

\begin{lemma}
\label{lem:goe_eps_subspace_dim}
	For $\G\sim \GOE(n)$ and any $\eps \in (0, 2)$, with probability $1-o(1)$,
	\begin{equation*}
		\frac{1}{n} \left|\left\{i: \lambda_i(\G) \geq (2-\eps)\sqrt{n}\right\}\right| \leq \frac{\eps^{3/2}}{\pi} (1 \pm o(1)).
	\end{equation*}
\end{lemma}
\begin{proof}
	We apply \pref{thm:semicircle_law} directly.
	The area under the semicircle between $2-\eps$ and 2 can be upper bounded by a rectangle of width $\eps$ and height $\sqrt{4\eps - \eps^2} \leq 2\sqrt{\eps}$.
	Dividing by $2\pi$ completes the proof.
\end{proof}

Next, we look at the adjacency matrix $\bA$ of a random $d$-regular graph.
It is a standard result that the largest eigenvalue of $\bA$ is $d$, and the all-ones vector $\vec{1}$ is the top eigenvector.
Thus, it is common to look at the ``de-meaned'' matrix $\meanbA \coloneqq \bA - \frac{d}{n}J$ (removing the top eigenvector component).
The eigenvalue bounds for $\meanbA$ were conjectured by Alon \cite{Alo86} and later proved by Friedman \cite{Fri08}.
The following quantitative statement is by Bordenave \cite{Bor15}.

\begin{theorem}
	\label{thm:eigenvalue_regular_graph}
	Let $d\geq 3$ be an integer and let $\bA$ be the adjacency matrix of a random $d$-regular graph on $n$ vertices.
	With probability $1- \frac{1}{\poly(n)}$, all eigenvalues of $\bA - \frac{d}{n}J$ lie within $[-2\sqrt{d-1}- \eps_n, 2\sqrt{d-1}+\eps_n]$,
	where $\eps_n = c \left(\frac{\log\log n}{\log n}\right)^2$ for some constant $c$.
\end{theorem}

For constant $d$ (fixed as $n$ grows), the empirical eigenvalue distribution is given by the \textit{Kesten--McKay law} \cite{Kes59,McK81}.
The eigenvalues of $\bA$ and $\bA - \frac{d}{n}J$ interlace by Cauchy's interlacing theorem, and hence the limiting eigenvalue distributions are the same for both matrices.

\begin{theorem}[Kesten--McKay law \cite{McK81}]
	\label{thm:kesten_mckay}
	Let $d \geq 3$ be a fixed integer.
	Let $\bA$ be the adjacency matrix of a random $d$-regular graph, and let $\meanbA = \bA - \frac{d}{n}J$.
	For any fixed real numbers $a < b$, with probability $1-o(1)$,
	\begin{equation*}
	\begin{gathered}
		\frac{1}{n} \left| \left\{i: \lambda_i(\meanbA) \in [a,b] \right\}\right| = (1\pm o(1)) \int_a^b \rho_{d}(x)dx, \\
		\rho_{d}(x) \coloneqq \frac{d\sqrt{4(d-1)-x^2}}{2\pi(d^2-x^2)} \quad \text{for } |x| \leq 2\sqrt{d-1}.
	\end{gathered}
	\end{equation*}
\end{theorem}

\begin{lemma}
\label{lem:km_eps_subspace_dim} 
	Let $d \geq 3$ be a fixed integer.
	Let $\bA$ be the adjacency matrix of a random $d$-regular graph, and let $\meanbA = \bA - \frac{d}{n}J$.
	For any $\eps \in (0,1)$, with probability $1-o(1)$,
	\begin{equation*}
		\frac{1}{n} \left|\left\{i: \lambda_i(\meanbA) \leq -2\sqrt{d-1}(1-\eps)\right\}\right| \leq \frac{12\sqrt{2}}{\pi}\eps^{3/2} (1 \pm o(1)).
	\end{equation*}
\end{lemma}
\begin{proof}
	We apply \pref{thm:kesten_mckay} directly.
	The Kesten--McKay density $\rho_d$ can be upper bounded by $\frac{d}{2\pi (d^2-4(d-1))} \sqrt{4(d-1)-x^2}$, a scaled semicircle of radius $2\sqrt{d-1}$.
	Let $x_\eps = 2\sqrt{d-1}(1-\eps)$.
	Then, $\rho_d(x_{\eps}) \leq \frac{d\sqrt{d-1}}{\pi(d-2)^2}\sqrt{2\eps}$.
	This upper bound is increasing with $\eps$.
	Thus, we can bound the area under $\rho_d$ between $2\sqrt{d-1}(1-\eps)$ and $2\sqrt{d-1}$ by $\frac{2\sqrt{2}}{\pi} \cdot \frac{d(d-1)}{(d-2)^2} \eps^{3/2}$.
	The term $\frac{d(d-1)}{(d-2)^2}$ is decreasing for $d\geq 3$ and the maximum is 6.
	This completes the proof.
\end{proof}

\section{The $\kXOR$ principle}
\label{sec:kSAT}

A crucial ingredient in our algorithms is that we can efficiently certify with high probability that any assignment that approximately satisfies a random $\kSAT$ formula $\bInst$ must also approximately satisfy the formula as $\kXOR$.  This generalizes the $k=3$ case that appears in \cite{Fei02,FO07} under the name ``$\ThreeXOR$-principle''.

Concretely, we show:
\begin{lemma}	\label{lem:kXOR-principle}
	Let $\bInst$ be a random $\kSAT$ formula on $m = \alpha n^{(k-1)/2}$ clauses.  There is an algorithm that with high probability certifies:
	\begin{displayquote}
		Any $(1-\eta)$-satisfying assignment of $\bInst$ must $\kXOR$-satisfy at least $(1-2^{k-1}\eta)m - 2^{O(k)}\sqrt{\alpha \log^5 n}\cdot n^{(k-1)/2}$ clauses.
	\end{displayquote}
\end{lemma}

To prove \pref{lem:kXOR-principle}, we will need the Fourier expansion of $\kSAT$.  Though this is well-known and simple to derive, we present a proof for completeness.
\begin{claim}	\label{claim:kSAT-fourier}
	$\kSAT(x_1,\dots,x_k) = 1 - 2^{-k} \sum_{\T\subseteq[k]} \prod_{i\in \T} x_i$.
\end{claim}
\begin{proof}
	First note that $\kSAT(x_1,\dots,x_k) = 1 - \Ind_{x=\vec{1}}$.  It remains to verify that all the Fourier coefficients of $\Ind_{x=\vec{1}}$ are equal to $2^{-k}$.  Indeed, for any $\T \subseteq [k]$:
	\[
		\wh{\Ind}_{x=\vec{1}} (\T) = \E_{\bx\sim\{\pm1\}^k}\chi_\T(\bx)\Ind_{x=\vec{1}}(\bx) = 2^{-k},
	\]
	which completes the proof.
\end{proof}

We are now ready to prove \pref{lem:kXOR-principle}.
\begin{proof}[Proof of \pref{lem:kXOR-principle}]
	If $x$ is a $(1-\eta)$-satisfying assignment for $\bInst$, then on one hand:
	\begin{align*}
		\frac{1}{|\bInst|} \sum_{(c,\cls)\in\bInst} \kSAT(c\circ x_\cls) \ge 1-\eta. \numberthis \label{eq:lower-bound-obj}
	\end{align*}
	On the other hand, by \pref{thm:aow-ref}, with high probability we can certify that $\bInst$ is quasirandom, specifically $\calD_{\bInst,x}$ is approximately $(k-1)$-wise uniform for all $x\in \{\pm1\}^n$ (recall \pref{def:eps_t_wise_uniform} and \pref{def:quasirandom}).
	Thus, we can certify:
	\begin{align*}
		\frac{1}{|\bInst|} \sum_{(c,\cls)\in\bInst} \kSAT(c\circ x_\cls) &= \E_{\bz\sim\{\pm1\}^k}\left[\kSAT(\bz) D_{\bInst,x}(\bz)\right] &\text{(by \pref{obs:obj-val})} \\
		&= (1-2^{-k})\wh{D_{\bInst,x}}(\varnothing) + 2^{-k}\sum_{\substack{\T\subseteq[k] \\ \T\ne\varnothing}}\wh{D_{\bInst,x}}(\T) &\text{(by \pref{claim:kSAT-fourier})} \\
		&= 1-2^{-k} + 2^{-k}\wh{D_{\bInst,x}}([k]) + 2^{-k}\sum_{\substack{\T\subseteq[k] \\ 1\le |\T|\le k-1}} \wh{D_{\bInst,x}}(\T) \\
		&\le 1-2^{-k} + 2^{-k}\wh{D_{\bInst,x}}([k]) + \frac{2^{O(k)}\log^{5/2}n}{\sqrt{\alpha}}. &\text{(by \pref{thm:aow-ref})}    \numberthis \label{eq:upper-bound-obj}
	\end{align*}
	If the certification of the above inequality fails, then we halt the algorithm and output \texttt{failure}.

	Otherwise, \pref{eq:lower-bound-obj} and \pref{eq:upper-bound-obj} together tell us that:
	\begin{align*}
		1-\eta &\le 1-2^{-k} + 2^{-k}\wh{D_{\bInst,x}}([k]) + \frac{2^{O(k)}\log^{5/2}n}{\sqrt{\alpha}},
	\end{align*}
	which can be rearranged as:
	\begin{align*}
		\wh{D_{\bInst,x}}([k]) \ge 1 - 2^k\eta - \frac{2^{O(k)}\log^{5/2}n}{\sqrt{\alpha}}.	\numberthis \label{eq:XOR-coeff-lower}
	\end{align*}
	By \pref{obs:obj-val} and \pref{eq:XOR-coeff-lower} the fraction of clauses of $\bInst$ that are $\kXOR$-satisfied by $x$ is:
	\[
		\frac{1+\wh{D_{\bInst,x}}([k])}{2} \ge 1 - 2^{k-1}\eta - \frac{2^{O(k)}\log^{5/2}n}{\sqrt{\alpha}},
	\]
	which completes the proof.
\end{proof}

\section{Count Certification for $k$CSPs}
\label{sec:kXOR}

Thanks to the $\kXOR$-principle, we can first focus on $\kXOR$.
For any $\kXOR$ instance, we can calculate the number of exactly satisfying assignments by Gaussian elimination.
However, the problem becomes non-trivial when we turn to the number of approximate solutions.
A priori, it is unclear whether we can certify a bound better than a naive $2^{O(n)}$ bound.
In this section, we will show an algorithm that certifies a subexponential upper bound when the underlying graph is random and sufficiently dense.

\subsection{Count certification for $\TwoXOR$}   \label{sec:2XOR}

As a warmup, we start with $\TwoXOR$.

\begin{theorem}
\label{thm:2xor_upper_bound}
    Let $\bG \sim \UnsignedH{2}{m}{n}$ be a random graph where $m = \Delta n$ and  $\Delta = n^{\delta}$ (for some constant $\delta > 0$).
    For any $\eta \in [0,1]$, there is a polynomial-time algorithm certifying that the number of $(1-\eta)$-satisfying assignments to any $\TwoXOR$ instance on $\bG$ is
    \begin{itemize}
        \item at most 2 if $\eta \leq \frac{1}{3n}$,
        \item at most $2e^{3\eta n \log n}$ if $\eta > \frac{1}{3n}$,
    \end{itemize}
    with probability at least $1 - \exp(-n^{\Omega(\delta)})$ over the randomness of $\bG$.
\end{theorem}

\begin{remark}
    Given a random graph, the algorithm simultaneously certifies an upper bound for all instances on this graph with arbitrary signings.
    The simultaneous certification will be crucial in the subsequent sections.
\end{remark}

Our first observation is that if the graph has large expansion, then given an approximate satisfying assignment, we cannot flip too many variables without violating many constraints.
Specifically, if we flip $S\subseteq [n]$ ($|S| < n/2$), then the number of clauses negated is $e(S, \overline{S})$, and this is large if the graph is an expander.
In fact, \pref{thm:2xor_upper_bound} holds for any instance with an expanding graph.
Our second observation is the following (which holds for all $k\geq 2$).

\begin{observation}
\label{obs:product_of_solns}
    Let $\Inst$ be a $\kXOR$ instance with arbitrary signings, and let $x, x'\in \{\pm1\}^n$ be two $(1-\eta)$-satisfying assignments.
    If a clause $(b,\cls)\in \Inst$ is satisfied by both, then $x_{\cls} \cdot x'_{\cls} = b^2 = 1$.
    Thus, the entry-wise product $x\circ x' \in \{\pm1\}^n$ is a $(1-2\eta)$-satisfying assignment for the all-positive instance $\pInst$ (changing all signings to $+1$).
\end{observation}

We now proceed to prove \pref{thm:2xor_upper_bound}.

\begin{proof}[Proof of \pref{thm:2xor_upper_bound}]
    The algorithm is as follows:
    given a graph with $n$ vertices and $m = \Delta n = n^{1+\delta}$ edges, and a parameter $\eta\in [0,1]$.

    \begin{enumerate}[(1)]
        \item Check that all vertex degrees are within $2\Delta \left(1\pm \frac{1}{\Delta^{1/3}}\right)$.
        \label{step:check_vertex_degree}

        \item Compute $\lambda_2(L)$ where $L$ is the normalized Laplacian matrix.
        Check that $\lambda_2(L) \geq 1- \frac{1}{\Delta^{1/4}}$.
        \label{step:check_spectral_gap}

        \item If the checks fail, output $2^n$.
        Otherwise, if $\eta \leq \frac{1}{3n}$, output 2;
        otherwise, output $2e^{3\eta n \log n}$.
    \end{enumerate}

    The random graph $\bG$ is an \ER random graph sampled from $\calG(n,p)$ with $p = \frac{2m}{n^2} = 2n^{-1+\delta}$ (removing parallel edges and self-loops allowed by the random model).
    It is a standard result in random graph theory that at this density, all vertex degrees concentrate around $np = 2\Delta$; specifically, the check in \pref{step:check_vertex_degree} will succeed with probability $1- \exp(-\Omega(\Delta^{1/3}))$ by a simple Chernoff bound (cf.\ \cite{FK16}).
    Furthermore, the check in \pref{step:check_spectral_gap} will succeed with probability $1-\exp(-\Omega(\Delta^{1/4}))$ due to \pref{thm:spectral_gap}.
    Note that for any instance where the checks fail, the output $2^n$ is still a valid (trivial) upper bound.

    Consider a maximum satisfying assignment $x$ (assume it is $(1-\eta)$-satisfying, otherwise our bound trivially holds), and let $x'$ be any $(1-\eta)$-satisfying assignment.
    By \pref{obs:product_of_solns}, $y \coloneqq x\circ x'$ is a $(1-2\eta)$-satisfying assignment to the all-positive instance $\pInst$ on $\bG$.
    Clearly, the all-ones vector $\vec{1}$ and its negation are exactly satisfying assignments to $\pInst$.
    We will show that $y$ must be close to $\vec{1}$ or $-\vec{1}$, meaning that $x'$ must be close to $x$ or $-x$ in Hamming distance.

    Suppose (without loss of generality) that $y$ is closer to $\vec{1}$, and let $S = \{i: y_i = -1\}$ where $|S| \leq \frac{n}{2}$.
    Then, the constraints of $\pInst$ between $S$ and $\ol{S}$ will be violated.
    By Cheeger's inequality (\pref{thm:cheeger}), $e(S,\ol{S}) \geq \frac{\lambda_2}{2} \vol(S)$, and by the degree concentration, $\vol(S) \geq |S| \cdot 2\Delta(1 - o(1))$.
    Thus, the spectral gap $\lambda_2\geq 1-o(1)$ and the degree bounds together certify that
    \begin{equation*}
        e(S,\ol{S}) \geq \Delta |S| (1-o(1)).
    \end{equation*}

    If $|S| \geq 3\eta n$, then $e(S,\ol{S}) > 2\eta \Delta n = 2\eta m$, contradicting that $y$ is a $(1-2\eta)$-satisfying assignment of $\pInst$.
    Thus, any $(1-\eta)$-satisfying assignment $x'$ must be $\floor{3\eta n}$-close to $x$ or $-x$ in Hamming distance.
    Note that if $\eta \leq \frac{1}{3n}$, then $x'$ can only be $\pm x$.

    For the $\eta > \frac{1}{3n}$ case, the number of assignments $\floor{3\eta n}$ away from $\pm x$ is upper bounded by
    \begin{equation*}
        2\sum_{\ell=0}^{\min(\floor{3\eta n}, n)} \binom{n}{\ell}
        \leq 2e^{3\eta n \log n}.
    \end{equation*}
    Note that if $\eta = O(1)$, then the upper bound trivially holds.
    For $\eta = o(1)$, we use the fact that $\binom{n}{3\eta n} \leq \frac{n^{3\eta n}}{(3\eta n)!}$.
\end{proof}

\subsection{Count certification for $\kXOR$}

For $k\geq 3$, we can obtain subexponential upper bounds by a recursive algorithm.
Same as \pref{thm:2xor_upper_bound}, our algorithm simultaneously certifies an upper bound for all $\kXOR$ instances on the given random hypergraph.

\begin{theorem}
\label{thm:kxor_upper_bound}
    For constant $k\geq 3$, let $\bH \sim \UnsignedH{k}{m}{n}$ be a random $k$-uniform hypergraph with $m = \Delta n$ and $\Delta = n^{\delta}$ (for some constant $\delta\in (0, k-1)$).
    For any $\eta \in [0,1]$ and $\eps > 0$, there is a polynomial-time algorithm certifying that the number of $(1-\eta)$-satisfying assignments to any $\kXOR$ instance on $\bH$ is at most
    \begin{equation*}
        2(2n)^{k-2} \cdot \exp\left(O(\eta n \log n) \right) \cdot \exp\left(O\left(n^{1 - \frac{\delta}{k-2}+ \eps}\right)\right)
    \end{equation*}
    with probability at least $1 - n^k\exp(-n^{\Omega(\eps)})$ over the randomness of $\bH$.
\end{theorem}

\begin{remark}
    The $(2n)^{k-2}$ in the upper bound is there to handle the case when $\eta = o(1/n)$ and $\delta \geq k-2$ (very dense instance).
    In this case, we get a $\poly(n)$ upper bound.
\end{remark}

The main idea is that if we have a certification algorithm for $\kmoXOR$, then we can obtain upper bounds for the \textit{induced $\kmoXOR$} defined in \pref{def:induced_txor}.
At a high level, the algorithm will do the following:
1) fix a set $S\subseteq V$ of a certain size,
2) look at the induced $(k-1)$-uniform hypergraph $H_{S}$ (\pref{def:induced_hypergraph}),
3) run the certification algorithm for $\kmoXOR$ on $H_{S}$ to obtain an upper bound,
and 4) multiply by $2^{|S|}$.

The intuition is that for every assignment $\sigma_S\in \{\pm1\}^S$, we get an induced $\kmoXOR$ instance with a random hypergraph and arbitrary signings (determined by $\sigma_S$).
Here we crucially use the fact that our algorithm simultaneously certifies a bound for all signing patterns, hence we avoid enumerating every assignment $\sigma_S$.
Once we have an upper bound on approximate solutions to the induced instances, we simply multiply it by $2^{|S|}$ to get the final upper bound.

We immediately see that for a fixed subset $S$, the above procedure throws away most of the clauses (keeping only clauses that have 1 variable in $S$).
Thus, it is clearly suboptimal to look at just one subset $S$.
To resolve this, we partition $V$ into subsets $S_1,\dots, S_\ell$, run the algorithm on each of them, and aggregate the results via the following lemma.

\begin{lemma}
\label{lem:average_partition}
    Given a $\kXOR$ instance $\Inst$, a partition of vertices $S_1, \dots, S_\ell$, and a threshold $t$.
    Suppose for each $i \in [\ell]$ and each induced $\kmoXOR$ instance $\Inst_{S_i, \sigma_{S_i}}$, the number of assignments (on variables $[n]\setminus S_i$) that violate at most $\floor{\frac{kt}{\ell}}$ constraints is upper bounded by $u_i$,
    then the number of assignments violating at most $t$ constraints in $\Inst$ is upper bounded by $\sum_{i=1}^\ell 2^{|S_i|} u_i$.
\end{lemma}

\begin{proof}
    Let $H$ be the underlying $k$-uniform hypergraph.
    Consider the induced $(k-1)$-uniform hypergraphs $H_{S_1},\dots H_{S_\ell}$.
    Each hyperedge in $H$ contributes at most $k$ hyperedges in the induced hypergraphs (i.e.\ the union of $H_{S_1},\dots, H_{S_\ell}$ will have at most $k$ copies of the same hyperedge).
    Thus, for any assignment $\sigma$ that violates $\leq t$ constraints in $\Inst$, there must be an $i\in [\ell]$ such that $\leq \floor{\frac{kt}{\ell}}$ constraints are violated in the induced $\kmoXOR$ instance $\Inst_{S_i, \sigma_{S_i}}$.

    Next, we bound the number of assignments violating at most $t$ constraints.
    \begin{equation*}
    \begin{aligned}
        \sum_{\sigma\in \{\pm1\}^n} \bone(\sigma \text{ violates $\leq t$ constraints in $\Inst$})
        &\leq \sum_{\sigma\in \{\pm1\}^n} \sum_{i=1}^\ell \bone\left(\sigma\text{ violates $\leq \left\lfloor \frac{kt}{\ell} \right\rfloor$ constraints in $\Inst_{S_i,\sigma_{S_i}}$} \right) \\
        &\leq \sum_{i=1}^\ell \left|\left\{\sigma: \sigma \text{ violates $\leq \left\lfloor\frac{kt}{\ell} \right\rfloor$ constraints in $\Inst_{S_i,\sigma_{S_i}}$}\right\}\right| \\
        &\leq \sum_{i=1}^\ell 2^{|S_i|} u_i.
    \end{aligned}
    \end{equation*}
    The second inequality follows by switching the two summations, and the final inequality holds because an upper bound $u_i$ for the induced $\kmoXOR$ instance implies an upper bound of $2^{|S_i|} u_i$ by enumerating all possible assignments to $S_i$.
\end{proof}

In the proof of \pref{thm:kxor_upper_bound}, we will partition $[n]$ into $\ell = n^{1-c}$ subsets of size $n^c$ for some $c<1$ chosen later.
The upper bounds $u_1,\dots,u_{\ell}$ will be obtained recursively and will roughly be the same with high probability.
Thus, by \pref{lem:average_partition}, we get an upper bound of $\ell 2^{n^c} u$, where we can choose $c$ to obtain the optimal result.

\begin{proof}[Proof of \pref{thm:kxor_upper_bound}]
    First, we can assume without loss of generality that $\eps < \frac{\delta}{k-2}$, otherwise the upper bound trivially holds.
    Our algorithm is the exact same procedure as \pref{lem:average_partition}.

    \noindent\rule{16cm}{0.4pt}

    \noindent{\bf Certification Algorithm for $\kXOR$}

    \noindent{\bf Input:}
    $k$-uniform hypergraph with $n$ variables and $m = \Delta n = n^{1+\delta}$ edges, parameters $\eta\in [0,1], \eps \in (0, \frac{\delta}{k-2})$.

    \begin{enumerate}[(1)]
        \item If $\delta < k-2$, choose $c = 1 - \frac{\delta}{k-2} + \eps$ for $0 < \eps < \frac{\delta}{k-2}$; if $\delta \geq k-2$, choose $c = 0$.
        Partition the vertices into $\ell = n^{1-c}$ subsets $S_1,\dots, S_{\ell}$ of size $n^c$, and extract the induced $(k-1)$-uniform hypergraphs $H_{S_1},\dots,H_{S_\ell}$ (removing duplicate hyperedges).

        \item Run the $\kmoXOR$ certification algorithm with $\eta' m' = \frac{k}{\ell}\eta m$ on each $H_{S_i}$, where $m'$ is the number of hyperedges in $H_{S_i}$.
        Let $u_i$ be the upper bound.

        \item Output $\sum_{i=1}^\ell 2^{|S_i|} u_i$.
    \end{enumerate}

    \noindent\rule{16cm}{0.4pt}

    We will prove the correctness of the algorithm by induction on $k$.
    The base case is $k=2$, and our $\twoxor$ algorithm from \pref{thm:2xor_upper_bound} achieves the same guarantees
    (we assume in this case $n^{1-\frac{\delta}{k-2}+\eps} = 0$).
    Now, suppose $k\geq 3$ and we have a $\kmoXOR$ algorithm with performance as stated in \pref{thm:kxor_upper_bound}.
    Then, by \pref{lem:average_partition}, the output is a valid upper bound on the number of assignments violating at most $\eta m$ constraints.
    \pref{thm:kxor_upper_bound} requires the number of hyperedges $m = n^{1+\delta}$ for some constant $\delta \in (0, k-1)$, and thus it suffices to prove that the induced hypergraphs $H_{S_1},\dots, H_{S_\ell}$ have the required density, which we can control by choosing $c$.

    For each $S_i$, the induced $\kmoXOR$ is a random $(k-1)$-uniform hypergraph where each $(k-1)$-tuple is included with probability $q \coloneqq 1- (1-p)^{|S_i|}$, where $p = \frac{m}{n^k} = n^{1+\delta-k}$
    (recall that we treat hyperedges as tuples; removing duplicates will not affect the upper bound).
    Here we split into two cases,
    \begin{itemize}
        \item $\delta \leq k-2$: we set $c = 1 - \frac{\delta}{k-2} + \eps < 1$.
        Then, $p |S_i| = n^{1+\delta-k}\cdot n^c < n^{-1+c} = o(1)$.
        Thus, $q = p|S_i| (1 \pm o(1)) = pn^c (1 \pm o(1))$.

        \item $\delta > k-2$: we set $c = 0$, thus $q = p = pn^c$.
    \end{itemize}
    In both cases, the number of hyperedges in the induced hypergraph is concentrated around
    \begin{equation*}
        m' = q n^{k-1} = \frac{m}{n^k} n^c (1 \pm o(1))\cdot n^{k-1} = n^{\delta + c} (1 \pm o(1)).
    \end{equation*}
    Thus, the density $\Delta' = \frac{m'}{n-n^c} \sim n^{\delta+c-1}$.
    Let $\delta' \coloneqq \delta + c - 1$.
    Again, we split into the two cases,
    \begin{itemize}
        \item $\delta \leq k-2$: we have $\delta' = \delta + (1-\frac{\delta}{k-2} + \eps) -1 = (1- \frac{1}{k-2})\delta + \eps \geq \eps > 0$ since $\eps > 0$, and $\delta' \leq (k-2) + c-1 < k-2$ since $c < 1$.
        Further, $\frac{\delta'}{k-3} \geq \frac{\delta}{k-2} + \frac{\eps}{k-3} > \frac{\delta}{k-2}$.

        \item $\delta > k-2$: we have $\delta' = \delta - 1 > k-3 \geq 0$ since $k \geq 3$, and $\delta' < k-2$ since $\delta < k-1$.
        Further, $\frac{\delta'}{k-3} = \frac{\delta-1}{k-3} > \frac{\delta}{k-2}$.
    \end{itemize}
    In both cases, the induced $\kmoXOR$ instance has the required density $\Delta' = n^{\delta'}$ with $\delta' \in (0, k-2)$, which means we can apply the $\kmoXOR$ algorithm.
    The parameter $\eta'$ is set to $\frac{1}{m'} \cdot \frac{k}{\ell} \eta m \sim k\eta$ (capped at 1), and set $\eps' = \eps$.

    The $\kmoXOR$ algorithm on the induced instance will certify an upper bound of
    \begin{equation*}
        u_i \leq 2(2n)^{k-3} \exp\left(O(\eta' n \log n) \right) \cdot \exp\left( O\left(n^{1 - \frac{\delta'}{k-3}+ \eps'}\right)\right).
    \end{equation*}

    Since $\eta' = O(\eta)$, $\eps' = \eps$, and $\frac{\delta'}{k-3} > \frac{\delta}{k-2}$, our final upper bound is
    \begin{equation*}
        \sum_{i=1}^\ell 2^{|S_i|}u_i
        \leq 2(2n)^{k-2} \cdot \exp\left(O(\eta n \log n) \right) \cdot \exp\left(O\left(n^{1 - \frac{\delta}{k-2}+ \eps}\right)\right).
    \end{equation*}

    Finally, we bound the failure probability.
    The $\kmoXOR$ algorithm fails with probability $< n^{k-1}\exp(-n^{\Omega(\eps')}) = n^{k-1}\exp(-n^{\Omega(\eps)})$.
    We union bound over the $\ell$ induced hypergraphs, we get the failure probability $< n^k \exp(-n^{\Omega(\eps)})$.
\end{proof}

\subsection{Count certification for all $\kCSP$s}
First, observe that as an immediate consequence of \pref{thm:kxor_upper_bound} and \pref{lem:kXOR-principle}, we have:
\begin{corollary}   \label{cor:kSAT-count}
    For constant $k\ge 3$, let $\bInst\sim\HDist{k}{n}{m}$ be a random signed hypergraph where $m = \Delta n = n^{1+\delta}$.  For every constant $\eps > 0$, there is an algorithm that certifies with high probability that the number of $(1-\eta)$-satisfying assignments to $\bInst$ as an instance of $\kSAT$ is at most
    \[
        \exp\left(\wt{O}(\eta n)\right)\cdot\exp\left(\wt{O}\left(n^{\frac{k+1}{4}-\frac{\delta}{2}}\right)\right) \cdot \exp\left(O\left(n^{1-\frac{\delta}{k-2}+\eps}\right)\right).
    \]
\end{corollary}
It is simple to upgrade the statement of \pref{cor:kSAT-count} from the case of $\kSAT$ to all $\kCSP$s.
\begin{corollary}   \label{cor:kCSP-to-kSAT}
    Let $P$ be any predicate not equal to the constant-$1$ function.  For constant $k\ge 3$, let $\bInst\sim\HDist{k}{n}{m}$ be a random signed hypergraph where $m = \Delta n = n^{1+\delta}$.  For every constant $\eps>0$ there is an algorithm that certifies with high probability that the number of $(1-\eta)$-satisfying assignments to $\bInst$ as an instance of $P$ is at most
    \[
        \exp\left(\wt{O}(\eta n)\right)\cdot\exp\left(\wt{O}\left(n^{\frac{k+1}{4}-\frac{\delta}{2}}\right)\right) \cdot \exp\left(O\left(n^{1-\frac{\delta}{k-2}+\eps}\right)\right).
    \]
\end{corollary}
\begin{proof}
    Let $z\in\{\pm1\}^k$ be any string not in the support of $P$.  Construct a new signed hypergraph $\bInst'\coloneqq \{(c\circ z,U): (c,U)\in\bInst\}$.  If $x$ is a $(1-\eta)$-satisfying assignment to $\bInst$ as an instance of $P$, then $x$ is also a $(1-\eta)$-satisfying assignment to $\bInst'$ as an instance of $\kSAT$.  Further, it is easy to see that $\bInst'$ is also distributed as $\HDist{k}{n}{m}$, so we can apply the algorithm from \pref{cor:kSAT-count} to certify a bound on the number of $(1-\eta)$-satisfying assignments to $\bInst'$ as an instance of $\kSAT$.  Consequently, we get the desired statement.
\end{proof}

\section{Count certification of Solution Clusters}
\label{sec:clustering}

In this section, we bound the number of clusters of satisfying assignments.
Due to the $\ThreeXOR$-principle, we first focus on clusters of $(1-\eta)$-satisfying assignments to random $\ThreeXOR$ instances.

\begin{theorem} \label{thm:3XOR_clusters}
    Consider a random 3-uniform hypergraph $\bH\sim \UnsignedH{3}{m}{n}$ where $m = \Delta n$ and $\Delta = n^{\delta}$ for some constant $\delta\in(0,2)$.
    Let $\eta \in [0,\eta_0]$ where $\eta_0$ is a universal constant, and let $\theta \coloneqq \max(2\eta, \Delta^{-\frac{1}{2}}\log n)$.
    There is a polynomial-time algorithm certifying that the $(1-\eta)$-satisfying assignments to any $\ThreeXOR$ instance on $\bH$ are covered by at most
    \begin{equation*}
        \exp(O(\theta^2 \log(1/\theta))n)
    \end{equation*}
    diameter-$(\theta n)$ clusters, with probability at least $1-\frac{1}{\poly(n)}$ over the randomness of $\bH$.
\end{theorem}

As an immediate corollary of \pref{thm:3XOR_clusters}, \pref{lem:kXOR-principle}, and the reduction in the proof of \pref{cor:kCSP-to-kSAT} we have:
\begin{corollary}
    \label{cor:3SAT_clusters}
    Let $P$ be any $3$-ary predicate not equal to the constant-$1$ predicate.  Let $\bInst\sim\HDist{3}{n}{m}$ be a random signed hypergraph where $m = \Delta n = n^{1+\delta}$ for some constant $\delta \in (0,2)$.
    Let $\eta \in [0,\eta_0]$ where $\eta_0$ is a universal constant, and let $\theta \coloneqq 8\eta + O(\sqrt{\frac{\log^5 n}{\Delta}})$.
    There is an algorithm that certifies with high probability that the $(1-\eta)$-satisfying assignments to $\bInst$ as a $P$-$\CSP$ instance are covered by at most
    \begin{equation*}
        \exp(O(\theta^2 \log(1/\theta))n)
    \end{equation*}
    diameter-$(\theta n)$ clusters.
\end{corollary}

Inspecting the proof of \pref{thm:2xor_upper_bound}, we see that it actually proves a stronger statement: for any pair of $(1-\eta)$-satisfying assignments $x,x'$ to the $\TwoXOR$ instance $\Inst$, $x'$ must be $(3\eta n)$-close to $x$ or $-x$ in Hamming distance.
The proof looks at the instance $\pInst$ (where all signs are set to $+1$) and the expansion of the graph.

The proof of \pref{thm:3XOR_clusters} will follow a similar path.
On a high level, we will first prove that $y \coloneqq x\circ x'$ must be either close to $\vec{1}$ or be roughly balanced, i.e.\ $x$ and $x'$ must have Hamming distance close to 0 or roughly $\frac{n}{2}$.
The main ingredient is \pref{lem:3_hypergraph_structure} which lets us certify an important structural result of random 3-uniform hypergraphs, allowing us to reason about the constraints violated by $y$ in $\pInst$.
\pref{lem:3_hypergraph_structure} will be a crucial step in \pref{sec:global-cardinality} as well.

The second ingredient is a result in coding theory.
Since the clusters are roughly $\frac{n}{2}$ apart, the number of clusters must be upper bounded by the cardinality of the largest $\eps$-balanced binary error-correcting code.
The best known upper bound is $2^{O(\eps^2\log(1/\eps))n}$, obtained by \cite{MRRW77} using linear programming techniques and also by \cite{Alo09} using an analysis of perturbed identity matrices.
This gives our final result.

We begin by proving an eigenvalue bound for the \textit{primal graph} of a random 3-uniform hypergraph (the graph obtained by adding a 3-clique for each hyperedge; see \pref{def:primal_graph}).
For simplicity, we will implicitly assume that all hyperedges with repeated vertices (allowed in our random model) are removed; this will not affect the results.

\begin{lemma}
    \label{lem:primal_graph_spectral_gap}
    Let $\bH\sim \UnsignedH{3}{m}{n}$ be a random 3-uniform hypergraph where $m = \Delta n$ and $\Delta = n^{\delta}$ (for constant $\delta \in (0,2)$), and let $\bA$ be the adjacency graph of the primal graph $\bprimalH$.
    Then, there is a constant $C$ such that with probability $1-\frac{1}{\poly(n)}$,
    \begin{equation*}
        \left\|\bA - \frac{6\Delta}{n}J\right\| \leq C\sqrt{\Delta \log n}.
    \end{equation*}
\end{lemma}

\begin{proof}
    The primal graph is a random graph such that for each tuple $\cls = (a,b,c)\in [n]^3$ (no repeated vertices), edges $(a,b),(b,c), (c,a)$ are included with probability $p \coloneqq \frac{m}{n^3} = \frac{\Delta}{n^2}$.
    Let $A_\cls$ be the adjacency matrix of the graph containing just the 3 edges.
    Then, the adjacency matrix $\bA = \sum_{\cls\in [n]^3} \bone(\cls\in \bH) A_\cls$.

    Define $\bS_\cls \coloneqq (\bone(\cls\in \bH) - p) A_\cls$ and
    \begin{equation*}
        \bS \coloneqq \sum_{\cls \in [n]^3} \bS_\cls = \bA - p \cdot 6(n-2) (J - \Id).
    \end{equation*}

    Clearly, $\E[\bS] = 0$ since $\E[\bS_{\cls}] = 0$, and $\E[\bS_\cls^2] = p(1-p) A_\cls^2$.
    Further, $\sum_{U} A_U^2 = 6(n-2) (J + (n-2)\Id)$, thus
    \begin{equation*}
        v \coloneqq \left\|\sum_{U\in [n]^3} \E[\bS_\cls^2]\right\| = p(1-p) \left\|\sum_{U\in [n]^3} A_U^2\right\| = p(1-p) \cdot 6(n-1)(2n-2) \leq 12\Delta.
    \end{equation*}

    Moreover, $\|\bS_\cls\| \leq \|A_\cls\| \leq 2$.
    Thus, by the matrix Bernstein inequality (\pref{thm:spec-norm-bound}) and assuming $\Delta = \omega(\log n)$, for a large enough constant $C$,
    \begin{equation*}
        \Pr\left[ \|\bS\| \geq C\sqrt{\Delta \log n} \right] \leq \frac{1}{\poly(n)}.
    \end{equation*}

    Next, $\bS = \bA - \frac{6\Delta(n-2)}{n^2}(J-\Id) = (\bA - \frac{6\Delta}{n}J) + \frac{12\Delta}{n^2}J - \frac{6\Delta(n-2)}{n^2}\Id$.
    The second and third terms have norm $O(\frac{\Delta}{n})$, negligible compared to $\sqrt{\Delta \log n}$ when $\Delta \ll n^2\log n$.
    Thus, by the triangle inequality,
    \begin{equation*}
        \left\|\bA - \frac{6\Delta}{n}J\right\| \leq C\sqrt{\Delta \log n}.
        \qedhere
    \end{equation*}
\end{proof}

\pref{lem:primal_graph_spectral_gap} allows us to apply the expander mixing lemma on the primal graph and prove the following structural result for random hypergraphs.

\begin{lemma}
    \label{lem:3_hypergraph_structure}
    Let $\bH\sim \UnsignedH{3}{m}{n}$ be a random 3-uniform hypergraph where $m = \Delta n$ and $\Delta = n^{\delta}$ (for constant $\delta \in (0,2)$).
    There is an algorithm certifying that for all subsets $S\subseteq [n]$ such that $|S| = (\frac{1}{2} + \gamma)n$ and $\gamma \in (0,\frac{1}{2})$, 
    \begin{enumerate}
        \item the number of hyperedges with 2 variables in $S$ and 1 in $\ol{S}$ is at least $3m(\gamma - 2\gamma^2) - O(n\sqrt{\Delta \log n})$,
        \item the number of hyperedges fully contained in $S$ is at least $m(\gamma + 2\gamma^2) - O(n\sqrt{\Delta \log n})$,
    \end{enumerate}
    with probability $1-\frac{1}{\poly(n)}$ over the randomness of $\bH$.
\end{lemma}

\begin{proof}
    We first look at the primal graph $\bprimalH$.
    The average degree is $\frac{6m}{n} = 6\Delta$.
    Let $\bA$ be the adjacency matrix, and let $\meanbA = \bA - \frac{6\Delta}{n}J$.
    The certification algorithm is simply checking that $\|\meanbA\| \leq C\sqrt{\Delta \log n}$ for some constant $C$.
    By \pref{lem:primal_graph_spectral_gap}, this will succeed with high probability.
    We proceed to prove that this is a valid certificate.

    We categorize the hyperedges of $\bH$ into 4 groups $T_0, T_1, T_2, T_3$, where $T_i$ is the set of hyperedges with $i$ variables in $S$ and $3-i$ in $\ol{S}$.
    We first lower bound $|T_2|$.

    By the expander mixing lemma (\pref{thm:expander_mixing_lemma}), the number of edges (of $\bprimalH$) between $S$ and $\ol{S}$ is
    \begin{equation*}
        e(S,\ol{S}) = \frac{6\Delta}{n} |S| (n-|S|) \pm C\sqrt{\Delta \log n} \sqrt{|S|(n-|S|)}.
    \end{equation*}

    Moreover, the number of edges within $\ol{S}$ (note that $e(\ol{S},\ol{S})$ double counts the edges)
    \begin{equation*}
        \frac{1}{2}e(\ol{S},\ol{S}) = \frac{3\Delta}{n} |\ol{S}|^2 \pm C\sqrt{\Delta \log n} |\ol{S}|.
    \end{equation*}

    Observe that the edges between $S$ and $\ol{S}$ must come from $T_1$ and $T_2$, each hyperedge contributing 2 edges: $2|T_1| + 2|T_2| = e(S,\ol{S})$.
    On the other hand, the the edges within $\ol{S}$ come from $T_0, T_1$, each hyperedge contributing 3 and 1 edges respectively: $3|T_0| + |T_1| = \frac{1}{2}e(\ol{S}, \ol{S})$.
    Thus, we have $|T_2| \geq \frac{1}{2}e(S,\ol{S}) - \frac{1}{2}e(\ol{S},\ol{S})$.
    For $|S| = (\frac{1}{2}+\gamma)n$,
    \begin{equation*}
    \begin{aligned}
        |T_2| &\geq 
        \frac{3\Delta}{n}\cdot \left( \left(\frac{1}{2}-\gamma\right) \left(\frac{1}{2}+\gamma\right) - \left(\frac{1}{2}-\gamma\right)^2\right) n^2 - O(n\sqrt{\Delta \log n}) \\
        &= 3m (\gamma - 2\gamma^2) - O(n\sqrt{\Delta \log n}).
    \end{aligned}
    \end{equation*}

    Next, we lower bound $|T_3|$.
    Similar to the derivations for $|T_2|$, we observe that $3|T_3| + |T_2| = \frac{1}{2} e(S,S)$ and $|T_2| \leq \frac{1}{2}e(S,\ol{S})$, hence $|T_3| \geq \frac{1}{6}(e(S,S) - e(S,\ol{S}))$.
    Similar calculations show that
    \begin{equation*}
        |T_3| \geq m (\gamma + 2\gamma^2) - O(n\sqrt{\Delta \log n}).
        \qedhere
    \end{equation*}
\end{proof}

Note that for small $\gamma$, the error term $O(n\sqrt{\Delta \log n})$ is negligible compared to $m(\gamma \pm 2\gamma^2)$ as long as $\gamma \gg \sqrt{\frac{\log n}{\Delta}}$.
Moreover, in the first claim, for $\gamma$ close to $\frac{1}{2}$ (say $\gamma = \frac{1}{2}-\gamma'$, $|S| = (1-\gamma')n$), the error term is negligible as long as $\gamma' \gg \sqrt{\frac{\log n}{\Delta}}$.

Next, we state a result by \cite{Alo09}, which was proved using an elegant argument about rank lower bounds of perturbed identity matrices.  Towards doing so, we define an \emph{$\eps$-balanced code of length-$n$} as a subset $L$ of $\{\pm1\}^n$ such that every pair of distinct $x,y\in L$ have Hamming distance in $\frac{1\pm\eps}{2} n$.
\begin{lemma}[{\cite[Proposition 4.1]{Alo09}}]
    \label{lem:balanced_code_bound}
    For any $\frac{1}{\sqrt{n}} \leq \eps < \frac{1}{2}$, the cardinality of any $\eps$-balanced code of length $n$ is at most $2^{c\eps^2 \log(1/\eps)n}$ for some absolute constant $c$.
\end{lemma}

Finally, we are ready to prove \pref{thm:3XOR_clusters}.

\begin{proof}[Proof of \pref{thm:3XOR_clusters}]
    Similar to the proof of \pref{thm:2xor_upper_bound}, we consider the instance $\pInst$.
    By \pref{obs:product_of_solns}, for any $(1-\eta)$-satisfying assignments $x,x'$, the product $y \coloneqq x \circ x'$ is a $(1-2\eta)$-satisfying assignment for $\pInst$.

    Let $S_+ = \{i: y_i = +1\}$ and $S_- = \ol{S_+}$.
    Assume $|S_+| = (\frac{1}{2}+\gamma) n$ for $\gamma > 0$ ($x,x'$ agree on more than half).
    Since all $\ThreeXOR$ clauses have sign $+1$ in $\pInst$, the clauses that have 2 variables in $S_+$ and 1 in $S_-$ must be violated.
    By \pref{lem:3_hypergraph_structure}, we can certify a lower bound of $3m(\gamma - 2\gamma^2)(1-o(1))$ of such clauses when $\omega(\sqrt{\frac{\log n}{\Delta}}) \leq \gamma \leq \frac{1}{2} - \omega(\sqrt{\frac{\log n}{\Delta}})$.
    Thus, we take $\theta \coloneqq \max(2\eta, \Delta^{-\frac{1}{2}}\log n)$.
    For a small enough $\eta$ ($\eta < 1/6$ suffices), we can certify that the number of violated constraints $3m (\gamma-2\gamma^2)(1-o(1)) > 2\eta m$ for all $\gamma \in [\theta, \frac{1}{2}-\theta]$.
    This shows that $|S_+|$ must be either $\geq (1-\theta)n$ or $\leq (\frac{1}{2}+\theta)n$.

    On the other hand, suppose $|S_-| = (\frac{1}{2}+\gamma)n$ for $\gamma > 0$ ($x,x'$ agree on less than half).
    The clauses contained in $S_-$ must be violated.
    Again, \pref{lem:3_hypergraph_structure} allows us to lower bound such clauses by $m(\gamma + 2\gamma^2)(1-o(1)) > 2\eta m$ for all $\gamma \in [\theta, \frac{1}{2}]$.
    This shows that $|S_-|$ must be $\leq (\frac{1}{2}+\theta)n$.

    Combining the results, we can certify that $x, x'$ must have Hamming distance $\leq \theta n$ or between $[(\frac{1}{2}-\theta)n, (\frac{1}{2}+\theta)n]$.
    Thus, the $(1-\eta)$-satisfying solutions form clusters of diameter $\theta n$, and the distance between any two clusters is $(\frac{1}{2}\pm \theta)n$.
    If we pick one assignment from each cluster, this gives a $(2\theta)$-balanced code.
    Thus, by \pref{lem:balanced_code_bound}, we can upper bound the number of clusters by
    \begin{equation*}
        \exp(O(\theta^2 \log(1/\theta))n).
    \end{equation*}
    This completes the proof.
\end{proof}

\section{Refuting CSPs under global cardinality constraints}	\label{sec:global-cardinality}
In this section, we give an algorithm to strongly refute random CSPs with global cardinality constraints, i.e., constraints of the form $\sum_i x_i \ge B$ well under the refutation threshold for appropriate values of $B$. 

\subsection{Refuting $3\CSP$s under global cardinality constraints}

For $\ThreeSAT$, there is a strong refutation algorithm using the random hypergraph structure result of \pref{lem:3_hypergraph_structure}, without requiring the $\ThreeXOR$-principle.  Via an identical reduction as the one in the proof of \pref{cor:kCSP-to-kSAT} the below statement extends to all $3\CSP$s.

\begin{theorem}
	\label{thm:3SAT_global_card}
	Given a $\ThreeSAT$ instance $\bInst\sim \HDist{3}{m}{n}$ where $m = \Delta n = n^{1+\delta}$ for constant $\delta > 0$,
	there is an efficient algorithm that certifies with high probability that $\bInst$ has no $\rho$-biased assignment which is $(1-\eta)$-satisfying where $\rho \gg \sqrt{\frac{\log n}{\Delta}}$ and $\eta = \rho/32$.
\end{theorem}

\begin{proof}
	Given a random $\ThreeSAT$ instance $\bInst$, we extract sub-instances $\bInst|_+$ and $\bInst|_-$ consisting of clauses with no negations and fully-negated clauses, respectively.
	Each sub-instance is a random 3-uniform hypergraph with density $\frac{1}{8}\Delta$.
	Consider an assignment $x\in \{\pm1\}^n$, and let $S_+ = \{i: x_i = 1\}$ and $S_- = \{i: x_i = -1\}$.
	Our main insight is that all hyperedges of $\bInst|_-$ contained in $S_+$ must be violated, and similarly all hyperedges of $\bInst|_+$ contained in $S_-$ must be violated.

	First, we consider the case when $|S_+| = \frac{1+\rho}{2}n$ for $\rho \gg \sqrt{\frac{\log n}{\Delta}}$.
	By \pref{lem:3_hypergraph_structure}, we can certify with high probability that there must be more than $\frac{1}{8}m \cdot \frac{\rho}{4}$ hyperedges of $\bInst|_-$ contained in $S_+$.
	Thus, for $\eta = \rho/16$, any assignment $x$ such that $|S_+| \geq \frac{1+\rho}{2}n$ cannot be $(1-\eta)$-satisfying.

	Similarly, consider the case when $|S_-| = \frac{1+\rho}{2}n$ for $\rho \gg \sqrt{\frac{\log n}{\Delta}}$.
	Again by \pref{lem:3_hypergraph_structure}, we can certify with high probability that there must be more than $\frac{1}{8}m \cdot \frac{\rho}{4}$ hyperedges of $\bInst_+$ contained in $S_-$.
	Thus, any assignment $x$ such that $|S_-| \geq \frac{1+\rho}{2}n$ cannot be $(1-\eta)$-satisfying.

	Therefore, this certifies that any $\rho$-biased assignment $x$ cannot be $(1-\eta)$-satisfying.
\end{proof}

\begin{remark} \label{rem:comparison_to_KOS}
	We compare our result to the result of \cite[Corollary C.2]{KOS18}.
	For random $\ThreeXOR$ with $m = n^{\frac{3}{2}-\eps}$ (under the refutation threshold), they showed that the Sum-of-Squares algorithm can certify that there is no $\rho$-biased exactly satisfying assignment when $\rho = \wt{\Omega}(n^{-\frac{1}{4}+\frac{\eps}{2}})$.
	Our algorithm matches this cardinality condition for $\rho$, and further extends to $(1-\Theta(\rho))$-satisfying assignments and to arbitrary 3CSPs.
\end{remark}

\subsection{The case of $\kCSP$s when $k\ge 4$}
In this section we give an algorithm to refute $\kCSP$s under global cardinality constraint when $k\ge 4$.  Our approach yields quantitatively weaker guarantees so we make no effort to optimize the tradeoff between refutation quality and the imbalance in the global cardinality constraint. 
Akin to our certified counting algorithms, we begin by first giving strong refutation algorithms for $\kXOR$, then use the $\kXOR$-principle to extend the algorithm to $\kSAT$, which then implies a strong refutation algorithm for every $\kCSP$.  
For $\kXOR$ we prove:

\begin{theorem}	\label{thm:main-kXOR-global-card}
	Let $k\ge4$. Given a $\kXOR$ instance $\bInst\sim\HDist{k}{m}{n}$ where $m\coloneqq n^{\frac{k-2}{2} + \beta}$, there is an efficient algorithm that certifies with high probability that $\bInst$ has no $2\rho$-biased assignment which is $(1-\eta)$-satisfying where $\rho\gg \frac{\log^6 n}{n^{\beta/(k-2)}}$ and
	\[
		\eta = \rho^{k-2}\left(\frac{\rho^2}{2} -
		\wt{O}\left( \frac{1}{\rho^{(k-2)/2} n^{(k-4)/4} n^{\beta/2}} \right) - \wt{O}\left(\frac{1}{\rho^{(k-2)/4}n^{\beta/2}}\right) - \frac{2}{n}\right).
	\]
\end{theorem}

Via the $\kXOR$ principle (\pref{lem:kXOR-principle}), and the arbitrary $\kCSP$-to-$\kSAT$ reduction in the proof of \pref{cor:kCSP-to-kSAT}, we get the following statement for refutation of $\kCSP$s under global cardinality constraints.
\begin{corollary}	\label{cor:main-kCSP-global-card}
	Let $\bInst\sim\HDist{k}{m}{n}$ where $m\coloneqq n^{\frac{k-1}{2}+\beta}$ and $\beta > 0$.  For any predicate $P$ not equal to the constant-$1$ function and any constant $\rho>0$, there is an efficient algorithm that certifies that $\bInst$ has no $2\rho$-biased assignment which $(1-\rho^k/2)$-satisfies $\bInst$ as a $P$-CSP instance.
\end{corollary}
While a quantitatively stronger statement than the above is true, we present the simplified version for ease of exposition.

Now we turn our attention to proving \pref{thm:main-kXOR-global-card}, and in service of which we prove two other lemmas as ingredients.

The high level idea for our $\kXOR$ strong-refutation algorithm is to pick some set of vertices $S\subseteq[n]$ of size $n^c$ where $0<c<1$ is an appropriately chosen constant.  Then for any assignment $y$ to the variables in $S$, there is an induced $\twoxor$ instance $\bInst_{S,y,2}$ on $[n]\setminus S$ that must be approximately satisfied (recall the definition of induced instances in \pref{def:induced_txor}).  The two steps of the algorithm are then to:
\begin{enumerate}
\item Certify that for any $y$ the induced $\twoxor$ instance is $(1/2+\eps)$-positive for some small $\eps$.
\item Simultaneously strongly-refute the family of all $\twoxor$ instances that are $(1/2+\eps)$-positive under a global cardinality constraint.
\end{enumerate}

We first describe the algorithm that certifies that the induced $\twoxor$ instance is $(1/2+\eps)$-positive.
\begin{lemma}	\label{lem:induced-two-xor-balanced}
	Given a $\kXOR$ instance $\bInst\sim\HDist{k}{m}{n}$ where $m\coloneqq n^{\frac{k-2}{2}+\beta}$, constant $c$ satisfying $c > 1 - \frac{2\beta}{k-2}$, and a fixed set of vertices $S$ of cardinality $n^c$, there is an efficient algorithm to certify with high probability:
	\begin{displayquote}
		For any $y\in\{\pm1\}^S$ the instance $\bInst_{S,y,2}$ is $\left(\frac{1}{2}+\poseps\right)$-positive for the following choice of $\poseps$.
		\begin{align*}
			\poseps = \wt{O}\left(\max\left\{1, n^{\frac{\beta}{2} - \frac{k-2}{4}}\right\} \sqrt{\frac{n^{(1-c)\frac{k-2}{2}}}{n^{\beta}}}\right).
		\end{align*}
	\end{displayquote}
\end{lemma}
\begin{proof}
	To put the statement we want to certify in a different way, we would like an algorithm that certifies:
	\[
		\frac{1}{\left|\bInst_{S,y,2}\right|}\sum_{(b,U)\in\bInst_{S,y,2}} b \le 2\eps.
	\]
	By \pref{obs:trunc-vs-induced}, this is equivalent to certifying the following for the truncated instance:
	\[
		\frac{1}{\left|\TruncbInst{S}{k-2}\right|}\sum_{(b,U)\in\TruncbInst{S}{k-2}} b\prod_{i\in U}y_i \le 2\eps,
	\]
	whose LHS can then be rewritten as:
	\begin{align*}
		\frac{1}{\left|\TruncbInst{S}{k-2}\right|}\sum_{U\in S^{k-2}} \prod_{i\in U}y_i \sum_{(b,U)\in\TruncbInst{S}{k-2}} b.
	\end{align*}
	We write $\bw_U$ to denote $\sum_{(b,U)\in\TruncbInst{S}{k-2}} b$.  Each $\bw_U$ is distributed as the sum of $O(n^2)$ random variables which are independent, bounded, centered and each nonzero with probability $O(n^{\beta-1-k/2})$.  Denote $\alpha$ as the expected number of nonzero terms in the sum defining $\bw_U$; its value is $Cn^{\beta+1-k/2}$ for some constant $C>0$.  By standard binomial concentration and Hoeffding's inequality, we know that with high probability $|\bw_U|\le B\coloneqq\max\{\log n, \sqrt{\alpha}\log n\}$ for all $U\in S^{k-2}$.
	Further, the probability that a given $\bw_U$ is nonzero is at most $p\coloneqq\min\{\alpha,1\}$.  Define $\wt{\bw}_U$ as the random variable $\bw_U\Ind[|\bw_U|\le B]$.
	\begin{enumerate}
		\item $\wt{\bw}_U$ is a centered random variable since $\bw_U$ is symmetric around $0$.
		\item $\wt{\bw}_U$ is supported on $[-B,B]$.
		\item The probability that $\wt{\bw}_U$ is nonzero is at most $p$.
	\end{enumerate}
	Hence by the certification algorithm of \cite{AOW15} (\pref{thm:aow-poly}), we can certify with high probability that for any $y\in\{\pm1\}^S$:
	\[
		\sum_{U\in S^{k-2}} \wt{\bw}_U \prod_{i\in U}y_i \le 2^{O(k)} B\max\{\sqrt{p}|S|^{3(k-2)/4}, |S|^{(k-2)/2}\}\log^{3/2}|S|.
	\]
	Since $\bw_U=\wt{\bw}_U$ for all $U\in S^{k-2}$ with high probability, our algorithm can verify this and also certify an identical upper bound on:
	\[
		\sum_{U\in S^{k-2}} \bw_U \prod_{i\in U}y_i \le 2^{O(k)} B\max\{\sqrt{p}|S|^{3(k-2)/4}, |S|^{(k-2)/2}\}\log^{3/2}|S|. \numberthis \label{eq:cert-upper}
	\]
	Plugging $|S|=n^c$ and $p = \min\{Cn^{\beta+1-k/2},1\}$ into \pref{eq:cert-upper} with $c > 1 - \frac{2\beta}{k-2}$, we have $\sqrt{p} |S|^{3(k-2)/4} > |S|^{(k-2)/2}$.  Moreover, since $\left|\TruncbInst{S}{k-2}\right|$ concentrates around $(1\pm o(1)) (\frac{|S|}{n})^{k-2} m$, we can certify with high probability that:
	\begin{align*}
		\frac{1}{\left|\TruncbInst{S}{k-2}\right|}\sum_{U\in S^{k-2}} \bw_U \prod_{i\in U}y_i &\le 2^{O(k)} B \sqrt{\frac{n^{(1-c)\frac{k-2}{2}}}{n^{\beta}} \log^3 n} \\
		&\le 2^{O(k)}\max\left\{1, n^{\frac{\beta}{2} - \frac{k-2}{4}}\right\} \sqrt{\frac{n^{(1-c)\frac{k-2}{2}}}{n^{\beta}}} \log^{5/2} n,
	\end{align*}
	since $B = \max\{1,\sqrt{\alpha}\} \log n = O(\max\{1, n^{\frac{\beta}{2} - \frac{k-2}{4}}\} \log n)$.
\end{proof}

We now describe the algorithm to \emph{simultaneously} refute the relevant family of $\twoxor$ instances.  The following definition more concretely describes the family of instances we are interested in.
\begin{definition}
	Given a multi-graph $G$, let $\calF(G,\eps)$ be the collection of all $\twoxor$ instances on $G$ that are $\left(\frac{1}{2}+\eps\right)$-positive.
\end{definition}

\begin{lemma}	\label{lem:simul-2XOR}
	Let $\bG$ be an $n$-vertex random multi-graph with average degree $\Delta\ge\log^2 n$ obtained by independently adding $\Binom\left(r,\frac{\Delta}{nr}\right)$ edges between $i$ and $j$ for every pair of distinct $i,j\in[n]$.  Further assume $r>\frac{2\Delta}{n}$.  There is an efficient algorithm that takes in $\bG$ as input and with high probability certifies that every $\calI\in\calF(\bG,\eps)$ has no $\rho$-biased assignment that is $\left(1-\frac{\rho^2}{2}+8\sqrt{\frac{\log n}{\Delta}}+\eps+\frac{2}{n}\right)$-satisfying.
\end{lemma}
\begin{proof}
	Let $x$ be a $\rho$-biased assignment and let $Y = \{v\in [n]: x_v = +1\}$ where $|Y| = (\frac{1}{2}+c)n$ and $|c| \geq \frac{\rho}{2}$.  To prove the statement our algorithm will certify a lower bound on the fraction of constraints that are violated in any $\calI\in\calF(\bG,\eps)$.  Towards doing so, we will first certify some lower bound $\gamma$ on the fraction of constraints within $Y$ or $\ol{Y}$.  Now, for every constraint $\{u,v\}$ where $u,v\in Y$ or $u,v\in\ol{Y}$, $x_ux_v = +1$.  Since the fraction of constraints $(b,\{u,v\})$ for which $b=+1$ is bounded by $\frac{1}{2}+\eps$, the fraction of violated constraints must be at least $\gamma-\left(\frac{1}{2}+\eps\right)$.

	The certification of the lower bound is much like the proof of the expander mixing lemma.  The number of edges that are contained within $Y$ or $\ol{Y}$ is accounted by:
	\begin{align*}
		2(|E(\bG[Y])| + |E(\bG[\ol{Y}])|) &= 1_Y^{\top}A_{\bG}1_Y + 1_{\ol{Y}}^{\top}A_{\bG}1_{\ol{Y}} \\
		&= 1_Y^{\top} \E[A_{\bG}] 1_Y + 1_{\ol{Y}}^{\top} \E[A_{\bG}] 1_{\ol{Y}} + 1_Y^{\top}(A_{\bG}-\E A_{\bG})1_Y + 1_{\ol{Y}}^{\top}(A_{\bG}-\E A_{\bG}) 1_{\ol{Y}} \\
		&\ge \frac{\Delta}{n}\left(|Y|^2 + |\ol{Y}|^2 \right) - \|A_{\bG}-\E A_{\bG}\|\cdot n - \Delta.
	\end{align*}
	With $|Y| = \left(\frac{1}{2}+c\right)n$ and $|c| \geq \frac{\rho}{2}$, we can rewrite the above inequality as:
	\begin{align*}
		2(|E(\bG[Y])| + |E(\bG[\ol{Y}])|) &\ge \Delta n\left(\frac{1}{2}+2c^2\right) - \|A_{\bG}-\E A_{\bG}\|\cdot n - \Delta \\
		&\ge \frac{1}{2}\Delta n\left(1+\rho^2\right) - \|A_{\bG}-\E A_{\bG}\|\cdot n - \Delta
	\end{align*}
	With the above bound in hand, our algorithm can certify a lower bound of $\gamma-\left(\frac{1}{2}+\eps\right)$ where
	\[
		\gamma = \frac{1}{2|E(\bG)|} \left(\frac{1}{2}\Delta n\left(1+\rho^2\right) - \|A_{\bG}-\E A_{\bG}\|\cdot n - \Delta\right).
	\]
	Next, to treat the factor involving $|E(\bG)|$ observe that:
	\[
		2|E(\bG)| = 1^{\top}\E A_{\bG}1 + 1^{\top}(A_{\bG}-\E A_{\bG})1 \le \Delta (n-1) + \|A_{\bG}-\E A_{\bG}\|\cdot n.
	\]
	To complete the proof, it suffices to give a high-probability upper bound on $\|A_{\bG}-\E A_{\bG}\|$.  For the bound on $\|A_{\bG}-\E A_{\bG}\|$ we use the matrix Bernstein inequality, as stated in \pref{thm:spec-norm-bound}.  Let $r\cdot K_n$ be the graph on $[n]$ with $r$ parallel edges between every pair of vertices.  Then $\bG$ can be thought of as the graph obtained by sampling each $e\in E(r\cdot K_n)$ independently with probability $\frac{\Delta}{rn}$.  For $e\in E(r\cdot K_n)$ define $A_e$ as the adjacency matrix of the single edge $e$.  Then:
	\[
		A_{\bG}-\E A_{\bG} = \sum_{e\in E(r\cdot K_n)} A_e \cdot \left(\Ind[e\in\bG] - \frac{\Delta}{rn}\right).
	\]
	Then the $v$ parameter from the statement of \pref{thm:spec-norm-bound} for the above sum of random matrices is then equal to:
	\[
		\left\| \sum_{e\in E(r\cdot K_n)} \Id_e \E\left[\left(\Ind[e\in\bG] - \frac{\Delta}{rn}\right)^2\right] \right\| \le \left\| \Delta\cdot\left(1-\frac{\Delta}{nr}\right)\cdot\Id \right\| \le \Delta.
	\]
	Thus, by \pref{thm:spec-norm-bound}:
	\[
		\|A_{\bG}-\E A_{\bG}\| \le 4\sqrt{\Delta\log n}
	\]
	except with probability at most $1/n^2$.  When this holds for large enough $n$ this implies:
	\[
		\gamma \ge \frac{1}{2} + \frac{\rho^2}{2} - 8\sqrt{\frac{\log n}{\Delta}} - \frac{2}{n}.
	\]
	This in turn implies that the lower bound certified on the fraction of constraints violated by $x$ is with high probability at least:
	\[
		\frac{\rho^2}{2} - 8\sqrt{\frac{\log n}{\Delta}} - \frac{2}{n} - \eps
	\]
	which completes the proof.
\end{proof}

Now we are ready to prove \pref{thm:main-kXOR-global-card}.
\begin{proof}[Proof of \pref{thm:main-kXOR-global-card}]
	If $n^{\beta}\ge n\log^6 n$, then we can use the algorithm of \cite{AOW15} as stated in \pref{thm:aow-ref} to prove our statement.  Thus, we assume $n^{\beta}\le n\log^6 n$.  Let $S$ be some set, say $\{1,\dots,\ell\}$, where $\ell\coloneqq n^c$ is chosen so that $n^c = \rho n$.  Hence, when $k\ge 4$, the value of $\poseps$ is $\wt{O}\left(\frac{1}{\rho^{(k-2)/4} n^{\beta/2}}\right)$.
	By \pref{lem:induced-two-xor-balanced} we can certify with high probability that simultaneously for all assignments $y$ to variables in $S$, the induced $\TwoXOR$ formula is $\bInst_{S,y,2}$ is $\left(\frac{1}{2}+\wt{O}\left(\frac{1}{\rho^{(k-2)/4} n^{\beta/2}}\right)\right)$-positive.
	The underlying graph in $\bInst_{S,y,2}$ remains the same as we vary $y$.
	The expected number of edges in this graph is $\rho^{k-2}n^{\frac{k-2}{2}+\beta}\gg n^{\frac{k}{2}-1}\log^6 n$ and hence the number of edges concentrates around its expectation.
	And thus, the average degree $\Delta$ is $\rho^{k-2}n^{\frac{k-4}{2}+\beta}\gg \log^6 n$.  Further, the underlying graph is distributed exactly the same as in the hypothesis of \pref{lem:simul-2XOR}.
	Thus, when $k\ge 4$, an application of \pref{lem:simul-2XOR} tells us that we can certify with high probability that any $\rho$-biased assignment $x\in\{\pm 1\}^{n\setminus S}$ must violate at least
	\[
		\frac{\rho^2}{2} -
		\wt{O}\left( \frac{1}{\rho^{(k-2)/2} n^{(k-4)/4} n^{\beta/2}} \right) - \wt{O}\left(\frac{1}{\rho^{(k-2)/4}n^{\beta/2}}\right) - \frac{2}{n}
	\]
	fraction of the constraints on the induced formula $\bInst_{S,x_S,2}$, and consequently must violate at least
	\[
		\rho^{k-2}\left(\frac{\rho^2}{2} -
		\wt{O}\left( \frac{1}{\rho^{(k-2)/2} n^{(k-4)/4} n^{\beta/2}} \right) - \wt{O}\left(\frac{1}{\rho^{(k-2)/4}n^{\beta/2}}\right) - \frac{2}{n}\right)
	\]
	fraction of the constraints in $\bInst$.  Thus, any assignment $x$ that avoids violating at least the above fraction of constraints must satisfy $\left|\sum_{i\in[n]\setminus S}x_i\right|\le \rho n$.  Since $|S| = \rho n$, $x$ must be $2\rho$-biased.
\end{proof}

\section{Dimension-based count certification}
\label{sec:subspace_count}

We begin by upper bounding the number of Boolean vectors close to an arbitrary linear subspace.

\begin{theorem}
\label{thm:subspace_count}
    Let $V$ be a linear subspace of dimension $\alpha n$ in $\R^n$ for some $\alpha \in (0,1)$.
    For any $\eps\in (0,1/4)$, the number of Boolean vectors in $\cube$ that are $\eps$-close to $V$ is upper bounded by $2^{(H_2(4\eps^2) + \alpha \log \frac{3}{\eps})n}$.
\end{theorem}

\begin{proof}
    Let $T$ be the set of vectors in $\cube$ that are $\eps$-close to $V$, and let $B_V \coloneqq B_1(0) \cap V$ be the unit ball in $V$.
    We take an $\eps$-net $\calN_\eps$ of $B_V$.
    Every $x\in T$ is $\eps$-close to a point in $B_V$ (namely $\Pi_V x$), so by the triangle inequality, every $x\in T$ is $2\eps$-close to a point in $\calN_\eps$.

    Next, we bound the number of vectors in any $\eps$-ball.

    \begin{claim}
    \label{lem:boolean_in_ball}
        For any $\eps \in (0,1/\sqrt{2})$ and vector $u\in \R^{n}$, there can be at most $2^{H_2(\eps^2)n}$ Boolean vectors in $\cube$ contained in the $\eps$-ball $B_{\eps}(u)$.
    \end{claim}

    \begin{proof}
        If there are no Boolean vectors in $B_{\eps}(u)$ we are done.  So assume there is a Boolean vector $x\in B_{\eps}(u)$.
        For any Boolean vector $x'$ at Hamming distance at least $\eps^2 n$:
        \[
            \|x-x'\|_2 \ge \sqrt{\eps^2 n \cdot \frac{4}{n}} = 2\eps,
        \]
        which means $x'\notin B_{\eps}(u)$ and any Boolean vector in $B_{\eps}(u)$ must be Hamming distance at most $\eps^2 n$ from $x$, of which there are:
        \begin{equation*}
            \sum_{i=0}^{\eps^2 n} \binom{n}{i} \leq 2^{H_2(\eps^2) n},
        \end{equation*}
        where the inequality follows from the assumption $\eps < 1/\sqrt{2}$.
    \end{proof}

    A standard volume argument shows that there exists an $\eps$-net with cardinality $|\calN_\eps| \leq (\frac{3}{\eps})^{\alpha n}$ (see, for example, \cite[Corollary 4.2.13]{Ver18}).
    Finally, we bound the cardinality of $T$. Since $T \subseteq \bigcup_{u\in \calN_{\eps}} B_{2\eps}(u)$,
    \begin{equation*}
        |T| \leq |\calN_{\eps}| \cdot 2^{H_2(4\eps^2)n} \leq 2^{(H_2(4\eps^2) + \alpha \log \frac{3}{\eps})n}.
        \qedhere
    \end{equation*}
\end{proof}

\begin{remark} \label{rem:subspace_count_tight}
    The upper bound of \pref{thm:subspace_count} is almost tight.
    For some small constants $\alpha, \eps > 0$, consider the subcube
    $T = \left\{\pm \frac{1}{\sqrt{n}}(\vec{1}, y): y\in \{\pm1 \}^{\alpha n}\right\} \subset \cube$,
    and let $V = \spn(T)$.
    Clearly, $|T| = 2^{\alpha n+1}$ and $\dim(V) = \alpha n+1$.
    For any $x\in T$, there are
    \[
        \sum_{i=1}^{\frac{\eps^2n}{4}} \binom{(1-\alpha)n}{i} \geq 2^{H_2(\frac{\eps^2}{4(1-\alpha)})(1-\alpha)n - O(\log n)}
    \]
    number of Boolean vectors $\eps$-close to $x$ and differ from $x$ in the first $(1-\alpha)n$ coordinates.
    Multiplied by $|T|$, the number of Boolean vectors $\eps$-close to $V$ is at least $2^{(H_2(\frac{\eps^2}{4(1-\alpha)}) + \Omega(\alpha)) n - O(\log n)}$.
    This shows that the exponent in the upper bound of \pref{thm:subspace_count} is tight up to a $\log(1/\eps)$ factor.
\end{remark}

The idea of bounding the number of structured vectors close to a subspace will be the main theme in the following sections.
Specifically, for a matrix $M\in \R^{n \times n}$, if $x^\top M x \approx \lambda_{\max}(M) \|x\|_2^2$, then $x$ must be close to the top eigenspace of $M$.
This allows us to apply \pref{thm:subspace_count} (or a variant of it for independent sets).
For a linear subspace $V$, we denote $\Pi_{V^\perp}$ as the projection matrix to the orthogonal subspace $V^\perp$, i.e.\ $\|\Pi_{V^\perp}x\|_2$ is the distance from $x$ to $V$.
We will first prove the following useful lemma.

\begin{lemma} \label{lem:distance_to_subspace}
    Let $M$ be a symmetric $n\times n$ matrix with eigenvalues $\lambda_1 \geq \lambda_2 \geq \cdots \geq \lambda_n$ and orthonormal eigenvectors $v_1,\dots, v_n$ such that $\lambda_1 > 0$.
    Further, let $V \coloneqq \spn\{v_i: \lambda_i \geq \lambda_1(1 - \delta)\}$ for some constant $\delta \in (0,1)$.
    Suppose $x\in \R^n$ satisfies $x^\top M x \geq \lambda_1 (1-\eta) \|x\|_2^2$ for some $\eta > 0$, then $\|\Pi_{V^\perp} x\|_2 \leq \sqrt{\frac{\eta}{\delta}}\|x\|_2$.
\end{lemma}
\begin{proof}
    Let $\alpha \coloneqq \frac{1}{n} \dim(V)$.
    Let $x = \sum_{i=1}^n \wh{x}_i v_i$ written in the eigenvector basis.
    Clearly, we have $\sum_{i=1}^n \wh{x}_i^2 = \|x\|_2^2$ and $\sum_{\alpha n + 1}^n \wh{x}_i^2 = \|\Pi_{V^\perp} x\|_2^2$.
    \begin{equation*}
        x^\top M x = \sum_{i=1}^n \lambda_i \wh{x}_i^2 \leq \sum_{i=1}^{\alpha n} \lambda_1 \wh{x}_i^2 + \sum_{i=\alpha n + 1}^n \lambda_1 (1-\delta) \wh{x}_i^2
        \leq \lambda_1 \|x\|_2^2 - \delta \lambda_1 \|\Pi_{V^\perp} x\|_2^2.
    \end{equation*}
    Along with $x^\top M x \geq \lambda_1 (1-\eta) \|x\|_2^2$, we conclude
    $\|\Pi_{V^\perp} x\|_2^2 \leq \frac{\eta}{\delta} \|x\|_2^2$.
\end{proof}

\subsection{Sherrington-Kirkpatrick}
\label{sec:SK}

Recall the Sherrington-Kirkpatrick (SK) problem:
given $\G$ sampled from $\GOE(n)$, compute
\begin{equation*}
    \OPT(\G) = \max_{x\in \{\pm1\}^n} x^\top \G x.
\end{equation*}

A simple spectral refutation algorithm gives a \textit{spectral bound} of $\OPT(\G) \leq (2+o(1))n^{3/2}$.
We will certify an upper bound on the number of assignments achieving value close to the spectral bound.

\begin{theorem}
    \label{thm:SK_upper_bound}
    Let $\G \sim \GOE(n)$.
    Given $\eta \in (0,\eta_0)$ for some universal constant $\eta_0$,
    there is an algorithm certifying with high probability that at most $2^{O(\eta^{3/5}\log\frac{1}{\eta})n}$ assignments $x\in \{\pm1\}^n$ satisfy $x^\top \G x \geq 2(1-\eta)n^{3/2}$.
\end{theorem}

\begin{proof}
    The algorithm is as follows.

    \begin{enumerate}[(1)]
        \item Choose $\delta = \eta^{2/5}$, and let $\eps = \sqrt{\frac{\eta}{\delta}} = \eta^{3/10}$.

        \item Compute the eigenvalues $\lambda_1 \geq \cdots \geq \lambda_n$ of $\G$, and compute $\alpha = \frac{1}{n} |\{i: \lambda_i \geq 2(1 - \delta)\sqrt{n}\}|$.
        Check that $|\frac{\lambda_1}{\sqrt{n}} - 2| < n^{-1/4}$; output $2^n$ if this fails.
        \label{step:check_goe_spectrum}

        \item Output $2^{(H_2(16\eps^2) + \alpha \log \frac{3}{\eps})n}$.
    \end{enumerate}

    Let $v_1,\dots, v_n$ be the corresponding eigenvectors of $\G$,
    and let $V_\delta \coloneqq \spn\{v_i: \lambda_i \geq 2(1-\delta)\sqrt{n}\}$ be the top eigenspace of dimension $\alpha n$.
    By the semicircle law (\pref{lem:goe_eps_subspace_dim}), with high probability $\alpha \leq O(\delta^{3/2})$.
    Moreover, the check in \pref{step:check_goe_spectrum} will succeed with high probability due to \pref{fact:norm_of_goe}, and thus $\lambda_1 \leq (2+o(1))\sqrt{n}$.

    Next, consider a normalized Boolean vector $y\in \cube$ such that $y^\top \G y \geq 2(1-\eta)\sqrt{n}$.
    By \pref{lem:distance_to_subspace}, we have
    $\|\Pi_{V_\delta^\perp} y\|_2 \leq \sqrt{\frac{\eta}{\delta}} + o(1) = \eps + o(1)$, i.e.\ $y$ is $2\eps$-close to $V_{\delta}$.

    By \pref{thm:subspace_count}, the number of $y\in \cube$ that are $2\eps$-close to a $\alpha n$-dimensional subspace is
    \begin{equation*}
        2^{(H_2(16\eps^2) + \alpha \log \frac{3}{\eps})n}.
    \end{equation*}

    Finally, we use the fact that $H_2(p) \leq 2p\log_2\frac{1}{p}$ for $p \leq \frac{1}{2}$.
    Thus, for small enough $\eta < \eta_0$, 
    $H_2(16\eps^2) \leq O(\eps^2\log\frac{1}{\eps})$.
    Since $\alpha \leq O(\delta^{3/2})$, our choice $\delta = \eta^{2/5}$ gives us an upper bound
    \begin{equation*}
        2^{O(\eta^{3/5}\log\frac{1}{\eta})n}.
    \end{equation*}
    This completes the proof.
\end{proof}

\subsection{Independent sets}
\label{sec:independent_set}

Recall that the best known \textit{certifiable} upper bound of the largest independent set size (the independence number) in a random $d$-regular graph is by the smallest eigenvalue of the
adjacency matrix (known as Hoffman's bound).
We first present a proof of the certification, which will give us some insights for the counting problem.

For a set $S\subseteq V$, let $1_S\in \zo^n$ be the indicator vector of $S$.
We will heavily use the ``centered'' vector $y_S\in \R^n$ defined as follows,
\begin{equation*}
    y_S = 1_S - \frac{\angles{1_S, \vec{1}}}{n} \vec{1}, \quad
    y_S(i) = \begin{cases}
        1 - \frac{|S|}{n} & i\in S, \\
        -\frac{|S|}{n} & i\notin S.
    \end{cases}
\end{equation*}
In words, $y_S$ is the projection of $1_S$ onto the subspace orthogonal to the all-ones vector.
Crucially, we have $\angles{y_S, \vec{1}} = 0$ and $\|y_S\|_2^2 = |S|(1-\frac{|S|}{n})$.

For an adjacency matrix $A$, let $\meanA \coloneqq A - \frac{d}{n} J$ be the ``de-meaned'' adjacency matrix, i.e., the matrix obtained by projecting away the Perron eigenvector.
We will mainly use $\meanA$ because its eigenvalues are well-distributed, whereas $A$ has an outlier eigenvalue $d$.
However, note that they have the same minimum eigenvalue: $\lambda_{\min}(A) = \lambda_{\min}(\meanA) > 0$.
The following lemma, widely known as Hoffman's bound (see also \cite{FO05}), relates the independence number to $-\lambda_{\min}(A)$.

\begin{lemma}[Certifiable upper bound on independence number]
\label{lem:max_ind_set_size}
    Let $G$ be a $d$-regular graph with $n$ vertices, let $A$ be the adjacency matrix, and let $\lambda \coloneqq -\lambda_{\min}(A)$.
    Suppose $S\subseteq [n]$ is an independent set, then
    \[
        |S| \leq \frac{\lambda}{d+\lambda} n.
    \]
\end{lemma}
\begin{proof}
    Since $S$ is an independent set, $1_S^{\top} A 1_S = 0$.
    Further, we have $1_S^\top J 1_S = |S|^2$.
    Thus,
    \begin{equation*}
        1_S^\top \left(\frac{d}{n}J - A\right) 1_S = \frac{d}{n}|S|^2.
    \end{equation*}
    Denote $\meanA = A - \frac{d}{n}J$.
    Since $\vec{1}$ is in the kernel of $\meanA$, by the definition of $y_S$, we have $1_S^\top (- \meanA) 1_S = y_S^\top(- \meanA) y_S$.
    Thus,
    \begin{equation*}
        \frac{d}{n} |S|^2 = y_S^\top (- \meanA) y_S 
        \leq \lambda \|y_S\|_2^2
        = \lambda \cdot \frac{|S|(n-|S|)}{n},
        \numberthis \label{eq:ind_set_size_bound}
    \end{equation*}
    where $\lambda = \lambda_{\max}(-\meanA) = -\lambda_{\min}(A)$ since $A$ and $\meanA$ have the same minimum eigenvalue.
    This gives us the upper bound.
\end{proof}

For random $d$-regular graphs, $\lambda \leq 2\sqrt{d-1} + o(1)$ due to Friedman's Theorem (\pref{thm:eigenvalue_regular_graph}).
Denote $\Ram \coloneqq \frac{2\sqrt{d-1}}{d}$ and $\Cindset \coloneqq \frac{\Ram}{1 + \Ram}$.
Note that for all $d\geq 2$, $\Ram \leq 1$, thus $\Cindset \leq \frac{1}{2}$.
\pref{lem:max_ind_set_size} allows us to certify with high probability that all independent sets in a random $d$-regular graph have size $\leq \Cindset n(1+o(1))$.

We then turn to the problem of counting large independent sets.

\begin{theorem}
    \label{thm:ind_set_upper_bound}
    Let $d\geq 3$ be a constant.
    For a random $d$-regular graph $\bG$ on $n$ vertices, given $\eta \in (0,\eta_0)$ for some universal constant $\eta_0$,
    there is an algorithm certifying with high probability that there are at most $2^{O(\eta^{3/5}\log\frac{1}{\eta})n}$ independent sets of size $\Cindset(1-\eta)n$.
\end{theorem}

\begin{remark}
    A trivial upper bound is $\binom{n}{\Cindset (1-\eta)n} \approx 2^{H_2(\Cindset(1-\eta))n} = 2^{\Omega_d(n)}$ for small $\eta$.
    Thus, for constant $d$ and small $\eta$, our upper bound is significantly better than this trivial bound.
\end{remark}

Let $\indsets{G}$ be the set of independent sets of size at least $\Cindset (1-\eta)n$ in a $d$-regular graph $G$, and let $\Y{G} = \{y_S: S\in \indsets{G}\}$.
Clearly, we have $|\indsets{G}| = |\Y{G}|$.
To bound $|\Y{\bG}|$ for a random $d$-regular graph $\bG$, we follow the same idea of bounding the number of vectors close to a subspace.
We first show that any $y_S\in \Y{\bG}$ is close to the top eigenspace of $-\meanbA$.

Let $v_1,\dots, v_n$ be the eigenvectors of $-\meanbA$, and let $\lambda_1,\dots,\lambda_n$ be the eigenvalues.
For a constant $\delta \in (0,1)$, let $V_\delta = \spn\{v_i: \lambda_i \geq d\Ram(1-\delta)\}$.

\begin{lemma}
    \label{lem:y_S_distance_to_subspace}
    Suppose $\lambda_{\max}(-\meanbA) \leq d\Ram (1+ o(1))$.
    Let $\eta \in (0,1)$.
    Then, for any independent set $S$ of size $\Cindset (1-\eta)n$, the vector $y_S$ satisfies
    $\|\Pi_{V_{\delta}^\perp} y_S\|_2 \leq \sqrt{\frac{2\eta}{\delta}}\|y_S\|_2$.
\end{lemma}
\begin{proof}
    Recall that $\|y_S\|_2^2 = |S| \left(1-\frac{|S|}{n}\right)$ and that $y_S^\top (-\meanbA) y_S = \frac{d}{n}|S|^2$.  Hence:
    \begin{align*}
        \frac{y_S^\top (-\meanbA) y_S }{\|y_S\|_2^2} &= \frac{d|S|}{n\left(1-\frac{|S|}{n}\right)} = \frac{dC_d(1-\eta)}{1-C_d(1-\eta)}
        = \frac{dr_d(1-\eta)}{1+\eta r_d}
        \ge dr_d(1-2\eta)
    \end{align*}
    where the last inequality uses $\frac{1}{1+t}\geq 1-t$ for $t\in (0,1)$ and $\Ram \leq 1$.

    Since $\lambda_{\max}(-\meanbA) \leq d\Ram (1+ o(1))$, the statement follows from \pref{lem:distance_to_subspace}.
\end{proof}

Next, similar to \pref{lem:boolean_in_ball}, we upper bound the number of $y_S\in \Y{\bG}$ that can be in the same $\eps\sqrt{n}$-ball.
Here, we have a factor of $\sqrt{n}$ in the radius because each $y_S \in \Y{\bG}$ has norm $\Theta(\sqrt{n})$.

\begin{lemma}\label{lem:y_in_ball}
    Let $\eps > 0$ such that $\eps <  \frac{1}{4\sqrt{2}}$, and let $G$ be a $d$-regular graph whose maximum independent set is bounded by $\frac{n}{2}$.
    There can be at most $2^{(32\eps^2 \log_2 \frac{1}{\eps})n}$ vectors in $\Y{G}$ contained in any $(\eps \sqrt{n})$-ball.
\end{lemma}
\begin{proof}
    For any sets $S,T\subseteq[n]$ and $|S|,|T|\le\frac{n}{2}$,
    we have $\|y_S - y_T\|_2^2 \ge |S\Delta T|/4$ where $S\Delta T$ is the symmetric difference.
    Thus, if $|S \Delta T| > 16\eps^2 n$, then they cannot be in the same $\eps \sqrt{n}$-ball.

    Pick any set $S$ for which $y_S\in\Y{G}$ (if no such set exists we are trivially done).  Then the number of sets $T$ such that $|S\Delta T| \le 16\eps^2 n$ is at most ${n\choose 16\eps^2 n}$, which is at most $2^{H_2(16\eps^2) n}$, which is at most $2^{(32\eps^2\log_2\frac{1}{\eps}) n}$ since $H_2(p)\le 2p\log_2\frac{1}{p}$ when $p\le\frac{1}{2}$. This completes the proof.
\end{proof}

Now, we are ready to prove \pref{thm:ind_set_upper_bound}.

\begin{proof}[Proof of \pref{thm:ind_set_upper_bound}]
    Recall that we defined $\Ram \coloneqq \frac{2\sqrt{d-1}}{d}$ and $\Cindset \coloneqq \frac{\Ram}{1+\Ram}$, and the de-meaned adjacency matrix $\meanbA = \bA - \frac{d}{n}J$.
    Further, define $\Cynorm \coloneqq \frac{\sqrt{\Ram}}{1+\Ram}$.
    Note that $\Cynorm < \sqrt{\Cindset} < 1$.
    The algorithm is as follows.

    \begin{enumerate}[(1)]
        \item Choose $\delta = \eta^{2/5}$, $\eps = \sqrt{\frac{2\eta}{\delta}}$, and $\eps'_d = 2\eps \Cynorm$.

        \item Compute the eigenvalues $\lambda_1 \geq \cdots \geq \lambda_n$ of $-\meanbA$, and compute $\alpha = \frac{1}{n} |\{i: \lambda_i \geq 2\sqrt{d-1}(1 - \delta)\}|$.
        Check that $|\lambda_1 - 2\sqrt{d-1}| < \frac{\log\log n}{\log n}$; output $2^n$ if this fails.
        \label{step:check_ramanujan}

        \item Output $2^{(32\eps'^2_d \log (1/\eps'_d) + \alpha \log(3/\eps)) n}$.
    \end{enumerate}

    Let $v_1,\dots, v_n$ be the eigenvectors of $-\meanbA$,
    and let $V_\delta \coloneqq \spn\{v_i: \lambda_i \geq 2\sqrt{d-1}(1-\delta)\}$ be the space spanned by the top $\alpha n$ eigenvectors.
    First, the check in \pref{step:check_ramanujan} will succeed with high probability due to Friedman's Theorem (\pref{thm:eigenvalue_regular_graph}).
    Thus, $\lambda_{\max}(-\meanbA) = 2\sqrt{d-1}(1+ o(1)) = d \Ram(1+o(1))$.
    Moreover, by the Kesten--McKay law (\pref{lem:km_eps_subspace_dim}), $\alpha \leq O(\delta^{3/2})$ with high probability.

    $\lambda_{\max}(-\meanbA)$ certifies that the maximum independent set has size at most $(1+o(1))\Cindset n$.
    Then, we have that every $y_S \in \Y{\bG}$ has norm $\|y_S\|_2^2 \leq (1+o(1))n \Cindset (1-\Cindset) = (1+o(1)) n \frac{\Ram}{(1+\Ram)^2} = (1+o(1)) \Cynorm^2 n$.
    Since $d\ge 3$, $(1+o(1))\Cindset\le \frac{1}{2}$ and by \pref{lem:y_S_distance_to_subspace}, every $y_S\in\Y{\bG}$ satisfies
    \begin{equation*}
        \|\Pi_{V_{\delta}^\perp} y\|_2 \leq \sqrt{\frac{2\eta}{\delta}} \|y_S\|_2 \leq \eps (1+o(1))\Cynorm \sqrt{n}.
    \end{equation*}
    Thus, $\Y{\bG}$ is contained in the centered ball of radius $2\Cynorm \sqrt{n}$ in $\R^n$, and is within distance $2\eps \Cynorm \sqrt{n}$ from the subspace $V_{\delta}$.

    We take an $\eps$-net $\calN_{\eps}$ of $B_1(0) \cap V_{\delta}$, the unit ball within $V_{\delta}$, and scale it by $2\Cynorm \sqrt{n}$.
    Then $\Y{\bG} \subseteq \bigcup_{u\in \calN_{\eps}} B_{\eps'_d \sqrt{n}}(u)$.
    By \pref{lem:y_in_ball}, each $(\eps'_d\sqrt{n})$-ball contains at most
    \begin{equation*}
        2^{(32\eps'^2_d \log \frac{1}{\eps'_d})n}
    \end{equation*}
    vectors in $\Y{\bG}$, provided that $\eps'_d < \frac{1}{4\sqrt{2}}$, which holds as long as $\eta < \eta_0$ for some universal constant $\eta_0$.

    The cardinality of $\calN_{\eps}$ is bounded by $\left(\frac{3}{\eps}\right)^{\alpha n} \leq 2^{(\alpha \log\frac{3}{\eps})n}$.
    Since $\alpha \leq O(\delta^{3/2})$ and $\Cynorm < 1$, our choice $\delta = \eta^{2/5}$ gives us an overall upper bound of
    \begin{equation*}
        2^{O(\eta^{3/5}\log\frac{1}{\eta})n},
    \end{equation*}
    which completes the proof.
\end{proof}

\section{Hardness evidence}
\label{sec:hardness}

In this section, we provide hardness evidence for some of the algorithmic problems we consider.  In these cases, improving slightly or significantly on our bounds would also result in improved algorithms for the refutation versions of these problems.  While we don't make confident claims of improvements on the refutation problems being computationally intractable, algorithms for them would certainly bypass several known barriers.

\subsection{Refutation-to-certified-counting reduction for $\kXOR$}
\label{sec:kXOR_hardness}

We first address the problem of counting solutions to a random $\kXOR$ instance.  Let $\bInst\sim\HDist{k}{m}{n}$ be a random $\kXOR$ instance on $m=\Delta n$ clauses.

\begin{theorem} \label{thm:kXOR_hardness}
    If there is an efficient algorithm that with high probability certifies a bound of $\exp(\frac{\eta n}{10k})$ on the number of $(1-\eta)$-satisfying assignments to $\bInst$, then there is an efficient algorithm that with high probability can certify that $\bInst$ has no $(1-\eta/2)$-satisfying assignments.
\end{theorem}
\begin{proof}
    Assume we ran the algorithm from the hypothesis of the theorem statement on $\bInst$, and obtained a bound on the count of $\exp(\frac{\eta n}{10k}) \leq 2^{\frac{\eta n}{4k}}$.  Take any set $S\subseteq[n]$ of size $\ell\coloneqq\frac{\eta n}{3k}$, say $\{1,\dots,\ell\}$.  If the number of clauses with at least one variable in $S$ is at most $\frac{\eta m}{2}$ (which holds with high probability), we output ``$\bInst$ has no $(1-\eta/2)$-satisfying assignments''.  Now, if there exists a $(1-\eta/2)$-satisfying assignment $x$ to $\bInst$, then any $x'$ which differs from $x$ on a subset of indices in $S$ must be $(1-\eta)$-satisfying.  However, since there are $2^{\frac{\eta n}{3k}}$ choices for $x'$ but a contradictory bound of $2^{\frac{\eta n}{4k}}$ on the number of such strings, there cannot be a $(1-\eta/2)$ satisfying assignment.
\end{proof}

\begin{remark} \label{rem:kXOR_count_tight}
    If $n^{\eps}\ll \Delta \ll n^{k/2-1}$, and $\eta n = \frac{n^{1+\eps}}{\Delta^{1/(k-2)}}$ for some $\eps > 0$, we are in a regime where:
    \begin{enumerate}
        \item there are no known algorithms to certify that there are no $(1-\eta/2)$-satisfying assignments, and
        \item our algorithm from \pref{thm:kxor_upper_bound} certifies a bound of $\exp(O(\eta n \log n))$.
    \end{enumerate}
    Thus, beating our algorithm in the above regime of $\eta$ by more than a logarithmic factor in the exponent would break a long-standing algorithmic barrier.
\end{remark}

\subsection{Refutation-to-certified-counting reduction for Independent Set}
\label{sec:ind_set_hardness}

In this section, we show evidence that our upper bound for independence number in a random $d$-regular graph (\pref{thm:ind_set_upper_bound}) cannot be improved significantly.

First, recall that we defined $\Ram \coloneqq \frac{2\sqrt{d-1}}{d}$ and $\Cindset \coloneqq \frac{\Ram}{1+\Ram}$.
Moreover, the best known \textit{certifiable} upper bound (Hoffman's bound) for the independence number of a random $d$-regular graph is $\Cindset n$.
It is widely believed that beating this bound (getting an upper bound of $(1-\eps)\Cindset n$ for some constant $\eps$) requires some new algorithmic ideas.

\begin{theorem} \label{thm:ind_set_hardness}
    Let $\bG$ be a random $d$-regular graph.
    Given constant $\eta \in (0, 1/2)$, if there is an efficient algorithm that with high probability certifies a bound of $\exp\left(\frac{\Cindset}{4} \eta \log(1/\eta)n \right)$ on the number of independent sets of size $\Cindset (1-\eta) n$, then there is an algorithm that with high probability certifies that $\bG$ has no independent set of size $(1-\eta/2)\Cindset n$.
\end{theorem}

\begin{proof}
    Assume we have an algorithm that obtains a bound of $\exp\left({\frac{\Cindset}{4}\eta \log(1/\eta) n}\right)$.
    Further, assume that there is an independent set $S$ of size $(1-\eta/2) \Cindset n$.
    Then, by the fact that any subset of $S$ is also an independent set, the number of independent sets of size $(1-\eta)\Cindset n$ must be at least
    \begin{equation*}
        \binom{(1-\frac{\eta}{2}) \Cindset n}{(1-\eta) \Cindset n} = \binom{(1-\eta/2)\Cindset n}{\frac{\eta}{2}\Cindset n}
        \geq 2^{H_2(\frac{\eta/2}{1-\eta/2})(1-\eta/2)\Cindset n - O(\log n)}
        >  \exp\left({\frac{\Cindset}{4}\eta \log(1/\eta) n}\right),
    \end{equation*}
    using the fact that $H_2(p) \geq p\log(1/p)$ and $\eta < 1/2$.
    This contradicts the upper bound.
\end{proof}

\begin{remark}
    For constant $d\geq 3$ and a small constant $\eta > 0$, \pref{thm:ind_set_upper_bound} certifies a bound of $\exp\left(O(\eta^{3/5}\log (1/\eta) n)\right)$.
    However, beating a bound of $\exp\left(O(\eta \log(1/\eta)n)\right)$ would also beat the Hoffman's bound by a factor of $(1-\eta/2)$.
    We conjecture that there is an efficient algorithm to certify an upper bound matching the lower bound, but we leave that as an open direction.
\end{remark}

\subsection{Approach for proving low-degree hardness for certified counts in $\kSAT$}   \label{sec:kSAT-hardness-approach}
We note that the counting hardness evidence for $\kXOR$ we provide is only evidence for optimality of our algorithms for counting \emph{approximately satisfying assignments} to CSPs.  It is desirable to give evidence suggesting that our algorithm's guarantees for certifying bounds on, say, the number of exactly satisfying assignments to a random $\kSAT$ formula are tight.  One approach for doing so is to construct a \emph{planted distribution} with an appropriately large number of $\kSAT$ assignments, and prove that it is impossible for low-degree polynomials to distinguish this planted distribution from random $\kSAT$ instances.  Here, we provide a blueprint for constructing such a planted distribution:
\begin{enumerate}
    \item Sample a random $k$-uniform hypergraph with a planted independent set $S$ (where $S$ is an independent set if no hyperedge contains $\ge 2$ vertices in $S$).
    \item Sample a random assignment $x$ on variables outside $S$.
    \item Place random negations $c$ so the clauses completely outside $S$ are satisfied by $x$ as $\ThreeXOR$, and for hyperedges $U=(u,v,w)$ with $u\in U$, the clause $(c_{U,v}v, c_{U,w}w)$ is satisfied as $\twoxor$.
\end{enumerate}
One of the challenges is in planting an independent set that doesn't ``stand out'' to spectral or low-degree polynomial distinguishers.  To this end we suspect that the techniques from \cite{BBKMW20} might be useful.

\section*{Acknowledgments}
We would like to thank Pravesh Kothari, Prasad Raghavendra, and Tselil Schramm for their encouragement and thorough feedback on an earlier draft.  We would also like to thank Tselil Schramm for enlightening discussions on refuting random CSPs.

\bibliographystyle{alpha}
\bibliography{main}

\end{document}